\documentclass[12pt]{article}     	 % Defines the document class

\usepackage[margin=1in]{geometry} % For nicer margins
\usepackage{comment}                % For commenting large portions of text
\usepackage[tbtags]{amsmath}                % Basic math package (includes align environment)
\usepackage{xcolor}                 % For colour (better than the color package)
\usepackage{enumitem}               % For lists (better than the enumerate package)
\usepackage{amssymb}                % Math symbols (includes boldface \mathbb{})
\usepackage{amsthm}                 % Enhances the \newtheorem command (includes the proof environment)
\usepackage{hyperref}               % For referencing (e.g. URLs)
\usepackage{graphicx}               % For including graphics (better than the graphics package)
\usepackage{tikz, pgfplots}         % TiKz is for drawing pictures, pgfplots makes some drawings easier
\usepackage{float}                  % Helpful for positioning figures
\usepackage{natbib}                 % Package used for easy citation
\usepackage{rotating}
\usepackage{setspace}
\usepackage{booktabs}
\usepackage{cleveref}
\usepackage{tcolorbox}
\usepackage{pdfpages}
\usepackage{bbm}
\usepackage[title]{appendix}

\theoremstyle{theorem}

\newtheorem{definition}{Definition}
\newtheorem{proposition}{Proposition}
\newtheorem{corollary}{Corollary}
\newtheorem{lemma}{Lemma}

\numberwithin{equation}{section}
	
\newcommand{\indep}{\perp \!\!\! \perp}
\title{% 
 	Identification of Impulse Response Functions for Nonlinear Dynamic Models\thanks{
 		We thank Burda, M., Lu, Y., Melino, A., Renault, E., the participants of the 2023 NBER/NSF Time Series Conference (Montreal), and of the 38th Canadian Econometrics Study Group Annual Meeting (Hamilton) for their helpful comments. This paper is a split from a previous version titled ``Nonlinear Impulse Response Functions and Local Projections", https://arxiv.org/abs/2305.18145. }}
\author{Gouri\'eroux, C.,\footnote{University of Toronto, Toulouse School of Economics and CREST, email: \textit{Christian.Gourieroux@ensae.fr}} and Q., Lee\footnote{University of Toronto, email: \textit{qt.lee@mail.utoronto.ca}}}
\date{\today, All comments welcome.}
\begin{document}
		\setstretch{1}
	%	\begin{minipage}{\textwidth}
	\maketitle
	\begin{abstract}
		\noindent We explore the issues of identification for nonlinear Impulse Response Functions in nonlinear dynamic models and discuss the settings in which the problem can be mitigated. In particular, we introduce the nonlinear autoregressive representation with Gaussian innovations and characterize the identified set. This set arises from the multiplicity of nonlinear innovations and transformations which leave invariant the standard normal density. We then discuss possible identifying restrictions, such as non-Gaussianity of independent sources, or identifiable parameters by means of learning algorithms, and the possibility of identification in nonlinear dynamic factor models when the underlying latent factors have different dynamics. We also explain how these identification results depend ultimately on the set of series under consideration. 
		 \\
		
		\noindent\textbf{Keywords:} Nonlinear Autoregressive Model, Generative Model, Impulse Response Functions, Nonlinear Independent Component Analysis, Local Projections, Nonlinear Innovations, Partial Identification, Recurrent Markov Process. \\
		\vspace{0in}\\
		\noindent\textbf{JEL Codes:} C01, C22.  \\
	\end{abstract}
	%	\end{minipage}

	\newpage
	
	\section{Introduction}
	
It is well known that in a  linear dynamic framework, such as the Gaussian Structural Vector Autoregressive (SVAR) model, there exists an identification issue for the structural innovations and respectively, the Impulse Response Function (IRF). However, when a nonlinear dynamic is introduced for the autoregressive process, it is possible for some of these problems to be mitigated. The objective in this paper is to highlight the key identification issues in this context and to discuss some possible remedies. \\

We begin by establishing the nonlinear autoregressive representation for a Markov process. This leads to the definition of nonlinear innovations and their potential interpretation in terms of shocks. In general however, the multivariate nonlinear autoregressive representation, the nonlinear innovations, the associated shocks and IRF are not defined in a unique way without additional parametric or nonparametric restrictions. We characterize the identification issue in this setting by describing the identified set (IS), and demonstrating that underidentification arises from transformations which preserve either the normal or uniform distribution on $[0,1]^n$. \\

There are a number of ways to resolve such identification issues in a nonlinear multivariate dynamic model. This arises if the innovations are non Gaussian, which facilitates identification based on the different power moments of the errors, or if the underlying dynamic independent sources have different nonlinear dynamics. In both cases, identification is conditional on the uniqueness of deconvolution. For instance, we can impose independence restrictions on non-Gaussian sources, leading to identifications by means of linear ICA, which facilitates identification based on the different power moments of the errors. Under partial identification, an appropriate use of algorithmic estimation (learning) approaches and of simulations can allow the recovery of identifiable (functional) parameters, such as the term structure of some Pseudo Impulse Response Functions (PIRF). Identification can also be facilitated in nonlinear dynamic factor models, when the underlying factors are independent with different dynamics. Howeverm it is important to highlight that these notions and key results ultimately depend on the universe, that is, on the set of series under consideration.\\

This paper is organized as follows. Section 2 describes the nonlinear autoregressive representation of a Markov process and presents the notion of nonlinear impulse response function in the nonlinear dynamic framework. Section 3 explains the identification issue arising from the multiplicity of nonlinear innovations. Section 4 provides methods in which the identification issue can be reduced. Section 5 discusses identification in the context of nonlinear dynamic factor models where the underlying factors feature different dynamics. Section 6 discusses the role of the universe in the context of identification, and Section 7 concludes. Proofs are gathered in the appendices and addition results are provided in the online appendices.

	\section{Nonlinear Framework and Impulse Response Function}
	
	\subsection{Nonlinear Autoregressive Representation of a Markov Process}

Let us first introduce the notion of the nonlinear autoregressive representation for a Markov process. We consider an $n$-dimensional Markov process  $(y_t$) of order 1 with values in $\mathcal{Y}=\mathbb{R}^n$. If its distribution is continuous, the Markov condition can be expressed in terms of its transition density, that is $f(y_t|\underline{y_{t-1}})=f(y_t|y_{t-1})$, $\forall y_t,\underline{y_{t-1}}$, where $\underline{y_{t-1}}=(y_{t-1},y_{t-2},...)$. Equivalently the Markov condition can be represented by the proposition below. 
\begin{proposition}
$(y_t)$ is a Markov process of order 1 on $\mathcal{Y}=\mathbb{R}^n$ with a strictly positive transition density: $f(y_t|y_{t-1})>0$, $\forall y_t,y_{t-1}$, if and only if it admits a nonlinear autoregressive representation:
\begin{equation}\label{nlar}
	y_t = g(y_{t-1};\varepsilon_t), \ t \geq 1,
\end{equation} 
where the $\varepsilon_t$'s are independent and identically distributed $N(0,Id)$ variables, with $\varepsilon_t$ being independent of $y_{t-1}$, and $g$ is a one-to-one transformation with respect to $\varepsilon_t$, that is continuously differentiable with a strictly positive Jacobian. The process $(\varepsilon_t)$ defines a Gaussian nonlinear innovation of the process $(y_t)$. 
\end{proposition} 

\textbf{Proof:} See Appendix A.1. \\

Nonlinear dynamic features can be introduced in this representation, since $y_t$ is a nonlinear function of its $y_{t-1}$ for a given $\varepsilon_t$, and/or a nonlinear function of $\varepsilon_t$ for a given $y_{t-1}$, and/or by nonlinear cross-effects of $y_{t-1}$ and $\varepsilon_{t}$. Moreover, such a representation exists even if $y_t$ has marginal and/or conditional fat tails. \\

The one-dimensional case has been discussed in Gouri\'eroux and Jasiak (2005, Section 2.1). Let us denote $F(\cdot|y_{t-1})$ as the conditional cumulative distribution function (c.d.f.) associated with $f(\cdot|y_{t-1})$. Under the assumption of Proposition 1, this function is invertible. We denote $Q(\cdot|y_{t-1})=F^{-1}(\cdot|y_{t-1})$ as its inverse, that is the conditional quantile function. Then we can choose:
\begin{equation}\label{fn_g}
	\varepsilon_t = \Phi^{-1} \circ F(y_t|y_{t-1}) \iff y_t = Q[\Phi(\varepsilon_t)|y_{t-1}] \equiv g(y_{t-1};\varepsilon_t),
\end{equation}
where $\Phi$ denotes the cdf of the standard normal. This is the only solution, if we impose that the one-to-one function $g$ is also increasing in $\varepsilon$ \footnote{Another solution with $g$ decreasing is obtained by changing $\varepsilon_t$ into $-\varepsilon_t$.}. Proposition 1 shows that this representation corresponds to the inversion method used for drawing in a distribution [Gourieroux and Monfort (1997)]. This approach can be extended to the multidimensional framework. In particular, we have the following corollary:
\begin{corollary}
	Any Markov process of order 1 on $y=\mathbb{R}^n$ can be written as:
	\begin{equation*}
		y_t = g(y_{t-1},u_t), \ t\geq1,
	\end{equation*}
	where the $u_t$'s are independent and identically distributed with uniform distribution on $[0,1]^n$. The process $(u_t)$ defines uniform nonlinear innovations. 
\end{corollary}
Proposition 1 and Corollary 1 are in particular valid when the data themselves are i.i.d., that is when we look for independent sources $\varepsilon_t$ or $u_t$ with effects on $y_t$. This also says that the conditions of Gaussianity or of uniform distributions are simply normalization conditions. More generally, we could have introduced a nonlinear innovation with a distribution other than $N(0,1)$ or $U[0,1]$ for the strict white noise components.\\

\begin{corollary}
	We can decompose any vector $y$ with continuous positive density as:
	\begin{equation*}
		\begin{split}
			y & = g(\varepsilon), \ \text{with} \ \varepsilon 
			\sim N(0,Id), \\ 
			\text{or} \ \		y & = \tilde{g}(u), \ \text{with} \ u 
			\sim U[0,1]^n. \\ 
		\end{split}
	\end{equation*}
\end{corollary}
This corresponds to nonlinear Independent Component Analysis (ICA). Corollary 2 shows that there is always a solution to the nonlinear ICA problem for continuous random vectors [see e.g. Hyvarinen and Pajunen (1999), Theorem 1]. 

\subsection{Effect of Time Unit}

The definition of the Gaussian nonlinear innovation in Proposition 1 (resp. of the uniform nonlinear innovation in Corollary 1) depends on the time unit. Let us discuss this by considering a time unit of 2. With this new time unit, the process $(y_{2\tau})$ is still a Markov process of order 1. Then we can write:
\begin{equation}
	y_{2\tau} = \tilde{g}(y_{2(\tau-1)},\tilde{\varepsilon}_{2\tau}), \ \text{$\tau$ varying,}
\end{equation}
where the $\tilde{\varepsilon}_{2\tau}$'s are i.i.d. $N(0,Id)$ variables, defining the Gaussian nonlinear innovations at step 2. Simultaneously, we can also write: 
\begin{equation}\label{36}
	y_{2\tau} = g^{(2)}(y_{2(\tau-1)},\varepsilon_{2(\tau-1)},\varepsilon_{2\tau}).
\end{equation}
\begin{enumerate}
	\item Let us consider a Gaussian AR(1) model: $y_t = \rho y_{t-1} + \varepsilon_t$, where the $\varepsilon_t$'s are i.i.d. $N(0,1)$. It is easily seen that $y_{2\tau} = \rho^2 y_{2(\tau-1)} + \varepsilon_{2\tau}+\rho \varepsilon_{2(\tau-1)}$. In this simple case, the function $\tilde{g}$ is still affine in $y_{2(\tau-1)}$ with an autoregressive coefficient $\rho^2$ and we have $\tilde{\varepsilon}_{2\tau} = \varepsilon_{2\tau} + \rho \varepsilon_{2(\tau-1)}$.
	\item  In the nonlinear dynamic framework, the nonlinear innovation $\tilde{\varepsilon}_{2\tau}$ is no longer a function of $\varepsilon_{2\tau}$ and $\varepsilon_{2(\tau-1)}$ only, but depends also on $y_{2(\tau-1)}$ in general. Therefore, the definition of nonlinear innovation depends in a complicated way on the selected time unit, with important consequences when considering the definition of shocks and the computation of IRFs. 
\end{enumerate}

\subsection{Nonlinear Impulse Response Function }

Our interest is to study the notion of shock and characterize its propagation mechanism in a nonlinear dynamic framework. It is important to discuss in this context the main assumptions in order to facilitate the introduction of shocks and the construction of their associated Impulse Response Function (IRF). For exposition, consider a bivariate framework with structural innovation $\varepsilon_t=(\varepsilon_{1,t},\varepsilon_{2,t}$). We consider ``shocks" of magnitude $\delta_1$ on the first component $\varepsilon_{1,t}$ (or respectively of $\delta_2$ on $\varepsilon_{2,t}$). To get reasonable interpretations, we assume that: 
\begin{enumerate}
	\item These shocks cannot change the past, which explains the need for the independence between $\underline{y_{t-1}}$ and $\varepsilon_t$.
	\item The shock on $\varepsilon_{1,t}$, say, has to be performed without an effect on $\varepsilon_{2,t}$. This explains the assumption of cross sectional independence between the components of $\varepsilon_t$. 
	\item It will be useful to fix the level of shocks $\delta_{1},\delta_{2}$ in a coherent way. This is done by imposing the conditions of identical distribution, where any level of shock corresponds to a given quantile of the common distribution. These common quantiles are what is used to define small as well as extreme shocks. 
\end{enumerate}

\textbf{Remark 1}: The literature has suggested an alternative definition of nonlinear innovation (in the one dimensional framework) as:
\begin{equation*}
	\varepsilon_t^* = \frac{y_t-\mathbb{E}(y_t|y_{t-1})}{\sqrt{\mathbb{V}(y_t|y_{t-1})}},
\end{equation*}
[see Blanchard and Quah (1989), Koop, Pesaran and Potter (1996)]. It is easily checked that $\varepsilon_t^*$ is not independent of $y_{t-1}$ in general and thus this definition is not appropriate for shocking $\varepsilon_t^*$ in the construction of IRFs. \\

The nonlinear IRF corresponding to a given nonlinear autoregressive representation is defined as a comparison between the perturbed $(y^\delta_{t+h})$ and baseline $(y_{t+h})$ future trajectories of the process. At horizon 0, a transitory shock of magnitude $\delta$ hits the perturbed path, which yields: 
\begin{equation*}
	y^{(\delta)}_{t} = g(y_{t-1},\varepsilon_{t}+\delta).
\end{equation*}
Then, in all subsequent horizons, the perturbed and baseline paths can be defined recursively as follows: 
\begin{equation}\label{perturbed}
	y^{(\delta)}_{t+h} = g(y^{(\delta)}_{t+h-1},\varepsilon_{t+h}),
\end{equation}
\begin{equation}\label{baseline}
	y_{t+h} = g(y_{t+h-1},\varepsilon_{t+h}),
\end{equation}
for all $h=1,2,...$. The corresponding IRF is given by:
\begin{equation}
	IRF_t(h,\delta,y_{t-1}) = y_{t+h}^{(\delta)} - y_{t+h}.
\end{equation}

This is a (nonlinear) function of the horizon $h$ (therefore, we have a term structure of the IRF), the magnitude of the shock $\delta$, the information set summarized by $y_{t-1}$ and also the sequence of future innovations $\varepsilon_t,...,\varepsilon_{t+h}$. Moreover, the IRF is stochastic conditional on $y_{t-1}$, due to the presence of future innovations. Therefore, we may partly summarize its distribution conditional on $y_{t-1}$ through the lens of a conditional expectation:
\begin{equation*}
EIRF(h,\delta,y_{t-1})=\mathbb{E}\left[IRF_t(h,\delta,y_{t-1})|y_{t-1}\right],
\end{equation*} $\forall \ h,\delta$, or a conditional covariance:
\begin{equation*}
CIRF(h,k,\delta,y_{t-1})=Cov\left[IRF_t(h,\delta,y_{t-1}),IRF_t(k,\delta,y_{t-1})|y_{t-1}\right],
\end{equation*} $\forall \ h,k,\delta$. \\

\textbf{Remark 2:} The IRF has a particular form in the special case of the linear Gaussian VAR framework\footnote{The Gaussian assumption is often implicit or explicit in the analysis of VAR(1) models [see e.g. Assumption 2 in Plagborg-Moller and Wolf (2021)].}:
\begin{equation}
	y_t = \Phi y_{t-1} + D \varepsilon_t.
\end{equation}
Then we get:
\begin{equation}
	IRF_t(h,\delta,y_{t-1}) = (Id + \Phi + ... + \Phi^h)D\delta,
	\end{equation}
which can be simplified to $(Id-\Phi)^{-1}(Id-\Phi^{k+1})D\delta$, if the autoregressive matrix $\Phi$ is such that $Id-\Phi$ is invertible. It is clear that the IRF in this context exhibits the following properties: 
\begin{enumerate}
	\item It is completely deterministic and not stochastic. 
	\item It is history invariant, that is, it does not depend on the current environment $y_{t-1}$. 
	\item It is a linear function with respect to the magnitude $\delta$. 
\end{enumerate}
Indeed, the Gaussian VAR(1) is a very special case of dynamic models and such properties are in general, not indicative of what arises in a nonlinear dynamic framework.

\subsection{Examples}

\textbf{Example 1: Conditionally Gaussian Model}\\

The model can be written as: 
\begin{equation}\label{38}
	y_t = m(y_{t-1}) + D(y_{t-1})\varepsilon_t,	
\end{equation}
where $\varepsilon_t\sim IIN(0,Id)$ and is assumed independent of $y_{t-1}$. This specification is used in Koop, Pesaran and Potter (1996) for instance. Except in the Gaussian linear VAR(1) framework: $m(y_{t-1})=A y_{t-1}, D(y_{t-1})=D$, independent of the past, the process is not conditionally Gaussian at horizons larger or equal to 2. At horizon 2, we get:
\begin{equation}
	y_{t+1} = m\left[m(y_{t-1})+D(y_{t-1})\varepsilon_t\right] + D(m(y_{t-1}+D(y_{t-1})\varepsilon_t)\varepsilon_{t+1},
\end{equation}
where the conditional distribution of $y_{t+1}$ given $y_{t-1}$ involves cross effects of $\varepsilon_t$, $\varepsilon_{t+1}$. Therefore, it is not Gaussian and it cannot be rewritten in the form:
\begin{equation}
	y_{t+1} = \tilde{m}(y_{t-1})+\tilde{D}(y_{t-1})\zeta_t,
\end{equation}
where $\zeta_t$ is Gaussian. \\

Equivalently, this specification is not invariant by a change of time unit. In the one dimensional case, the specification \eqref{38} includes the Double Autoregressive (DAR) model of order one [Weiss (1984), Borkovec and Kluppelberg (2001), Ling (2007)], introduced to account for conditional heteroscedasticity and given by:
\begin{equation}
	y_t = \gamma y_{t-1} + \sqrt{\alpha + \beta y_{t-1}^2} \ \varepsilon_t, \alpha >0,\beta\geq0,
\end{equation}
where $\varepsilon_t$ is $IIN(0,1)$. The DAR model has a strictly stationary solution if the Lyapunov coefficient is negative: 
\begin{equation}
	\mathbb{E}\log |\gamma + \sqrt{\beta} \ \varepsilon|<0.
\end{equation}
The conditional heteroscedasticity can create fat tails for the stationary distribution of process $(y_t)$. This process has second-order moments if morever:
\begin{equation*}
	\gamma^2 + \beta < 1.
\end{equation*}
When $	\gamma^2 + \beta > 1$ and $\mathbb{E}\log |\gamma + \sqrt{\beta} \ \varepsilon|Z<0$, the process $(y_t)$ is stationary with infinite marginal variance (but finite conditional variance). \\ 

Its multivariate extension can be written as:
\begin{equation*}
	y_t = \Phi y_{t-1} + (\text{diag}\ h_t)^{1/2}\varepsilon_t,
\end{equation*}
with $h_t=a + B(y^2_{1,t-1},...,y^2_{n,t-1})$ [see Zhu et al. (2017)]. \\

\textbf{Example 2: Qualitative Threshold AR(1) Model} \\

Such a model is defined by:
\begin{equation}
	y_t = \alpha \textbf{1}_{y_{t-1}>0} + \varepsilon_t,
\end{equation}
with $\varepsilon_t$ i.i.d. $N(0,\sigma^2)$. Then at horizon 2 we get: 
\begin{equation}
	\begin{split}
		y_{t+1} & = \alpha \textbf{1}_{y_t>0} + \varepsilon_{t+1}\\
		& = \alpha \textbf{1}_{\left\{\textbf{1}_{y_{t-1}>0} + \varepsilon_{t}>0\right\}} + \varepsilon_{t+1}\\
		& = \alpha \left[\textbf{1}_{y_{t-1}>0}\textbf{1}_{\alpha + \varepsilon_{t}>0}+\textbf{1}_{y_{t-1}<0}\textbf{1}_{\varepsilon_t>0} \right]\varepsilon_{t+1},\\
	\end{split}
\end{equation}

\textbf{Example 3: Time Discretized Diffusion Process}\\

Due to the lack of coherency of the definitions with respect to the time unit, it is sometimes proposed to write the structural definition of innovations and/or impulse response functions in the infinitesimal time unit even if the observations are available in discrete time. This is a usual practice in mathematical finance or in engineering. The discrete time process is $y_t=y(t),t=1,2,...$, where the underlying process $y(\tau)$ is defined in continuous time $\tau \in (0,\infty)$ by a multivariate diffusion equation:

\begin{equation}\label{diffusion}
	dy(\tau)=m[y(\tau)]d\tau + D[y(\tau)]dW(\tau),
\end{equation}
where $W$ is a multivariate Brownian motion with $\mathbb{V}[dW(\tau)]=Id. d\tau$. 
The diffusion equation \eqref{diffusion} is the analogue of the conditionally Gaussian model written in an infinitesimal time unit. For $n=1$, this class of time discretized diffusion contains the Gaussian AR(1) process, the autoregressive gamma (ARG) process, that is the time discretized Cox, Ingersoll, Ross process [Cox, Ingersoll and Ross (1985)], or the time discretized Jacobi process [Karlin and Taylor (1981), Gour\'ieroux, Jasiak and Sufana (2009) for its extension to the multivariate framework], with values in $(-\infty,+\infty),(0,+\infty),(0,1)$, respectively.

\section{The Identification Issue}

Except in irregular cases, the normalized nonlinear autoregressive representation is unique in the one-dimensional framework $n=1$ [see Appendix A.2. for an irregular case]. However, the uniqueness disappears in a multivariate setting. More precisely, let us consider the nonlinear autoregressive representation with uniform normalization and assume a true representation with true $g_0$, $u_t^0$. Then the identified set for $g$, $u_t$, is:
\begin{equation*}
	IS(g,u) \doteq \left\{g,u_t,\ \text{and} \ u_t \ \text{i.i.d} \ U[0,1]^n, \ \text{such that:} \ g(y_{t-1},u_t)\overset{d}{=}g_0(y_{t-1},u^0_t), \ \forall y_{t-1}\right\},
\end{equation*}
where $\overset{d}{=}$ means that the two variables have the same distribution. Since the functions $g_0$ and $g$ are invertible with respect to $u$, it follows that the identification issue arises from the possibility of having nonlinear transformations that preserve the uniform distribution on $[0,1]^n$.\footnote{The same reasoning can be applied with another normalization as the Gaussian one (see the proof in Section 3.2.)} This identification issue is linked to the identification issue in nonlinear Independent Component Analysis (ICA), where the literature considers the uniqueness of sources $u$ leading to a given output $y: \ y=g(u)$ [see Darmois (1953), Hyvarinen and Oja (2001), Roberts and Everson (2001) for the principles and practices of ICA.]. We are in the same framework with identical sources with varying information, that is, values of $y_{t-1}$. 

\subsection{Transformations Preserving the (Multivariate) Uniform Distribution}

\begin{proposition}
	Let us denote $u_t=(u_{1,t},...,u_{n,t})'$, where $u_{j,t}=\Phi(\varepsilon_{j,t})$, $j=1,...,n$, the uniform shocks associated with the Gaussian shocks. The uniform shocks are identified up to a continuously differentiable transformation $T$ from $[0,1]^n$ to $(0,1]^n$, which is invertible with a continuously differentiable inverse and such that the Jacobian satifies $\det\left(\frac{\partial T(u)}{\partial u'}\right)=+1, \ \forall u \in [0,1]^n$, or $\det\left(\frac{\partial T(u)}{\partial u'}\right)=-1, \ \forall u \in ]0,1[^n$.
\end{proposition}

\textbf{Proof:} See Appendix A.3. \\

To summarize, for $n=1$, both the nonlinear autoregressive representation $g$ and the Gaussian nonlinear innovation $\varepsilon_t$ are identifiable. This is no longer the case for $n\geq2$, where there is a problem of nonparametric partial identification.  Propositions 1 and 2 are valid under a property of null recurrence in the Markov process [see Tweedie (1975) for the definition of null recurrence], where, for any $y_{t-1}$, the support of the transition is $\mathcal{Y}$. Therefore, it is possible for a ``return" to any region of $\mathcal{Y}$ from any value $y_{t-1}$. This recurrence property is compatible with both stationary and nonstationary features of process $(y_t)$ (see the discussion in Section 3). We also make no assumptions about the observability of $(y_t)$. Thus, the propositions are valid for instance in a nonlinear state space model with latent (Markov) factors. Moreover, the assumption $\mathcal{Y}=\mathbb{R}^n$ is not restrictive. Indeed, the Markov property in Proposition 1 can be equivalently written on $y=(y_t)$, or a transform $\tilde{y}=(\tilde{y}_t)$ with $\tilde{y}_t = g^*(y_t)$, where $g^*(y_t)$ is a diffeomorphism from $\mathbb{R}^n$ to another set $\mathcal{Y}^*=g^*(\mathbb{R}^n)$. After this transformation, we get:
\begin{equation*}
	\tilde{y}_t = g^*[g[g^{*-1}(\tilde{y}_{t-1});\varepsilon_t]] \equiv \tilde{g}(\tilde{y}_{t-1};\varepsilon_t),
\end{equation*}
with the same shock $\varepsilon_t$. \\

\subsection{The Multiplicity of Nonlinear Innovations }

Let us now illustrate the multiplicity of Gaussian (resp. uniform) nonlinear innovations. It is well known that the linear transformation $D\varepsilon_t$ of the standard normal is also standard normal if and only if $DD'=Id$, that is, if $D$ is an orthogonal matrix. However, there exists also nonlinear transformations of $\varepsilon_t$ leaving invariant the standard normal distribution. 

\subsubsection{Preliminary Lemmas}

We see in Proposition 2 the key roles of the Jacobian matrices and their determinants. To apply the Jacobian formula, the Jacobian matrices must be invertible and matrices with strictly positive determinants form a multiplicative group. This group can be decomposed by subgroups with interesting interpretations. In particular, we have the following subgroups: 
\begin{itemize}
	\item The group of scalar positive matrices $\{\lambda Id, \lambda > 0\}$.
	\item The special linear group $SL(n)$ of matrices with determinant 1. 
	\item The special orthogonal group $SO(n)$ of orthogonal matrices with determinant 1\footnote{Q is an orthogonal matrix if $QQ'=Id$. In particular, $\det (QQ') = \det Q \det Q'= (\det Q)^2=\det Id = 1$, that is $\det Q = \pm 1$.}.
\end{itemize}
Then any matrix with positive determinant can be decomposed by considering successive quotient groups. We have the polar decomposition [see Hall (2015), Section 2.5]:
\begin{lemma}
	Any $n$ $\times$ $n$ matrix $A$ with strictly positive determinant can be written as:
	\begin{equation}
		A = \lambda Q \Omega,
	\end{equation}
	where $\lambda>0$, $Q$ belongs to $SO(n)$, $\Omega$ is symmetric positive definite with $\det \Omega = 1$.
\end{lemma}

 This decomposition can be also written by considering the matrix logarithm transforms. We denote $\text{Exp} \ A = \sum_{k=0}^\infty \frac{A^k}{k!}$. Then, we have the following lemma:

\begin{lemma}
	\begin{enumerate}
		\item 	The element Q of $SO(n)$ can be written as $Q=\text{Exp}\ B$, where $B$ is a skew symmetric matrix that satisfies $B'=-B$.
		\item The $\Omega$ symmetric positive definite matrix with $\det \Omega =1$ can be written as $\Omega= \text{Exp} \ V$, where $V$ is a symmetric traceless matrix (with $Tr \ V = 0$).
	\end{enumerate}
\end{lemma}

Note that the ``dimension" of the group generated by the different components are $1$ for $\lambda$, $\frac{n(n-1)}{2}$ for $Q$ (or $B$), $\frac{n(n+1)}{2}-1$ for $\Omega$ (or $V$). These dimensions sum up to $n^2$, that is, the dimension for the group of matrices $A$ with positive determinants. These decompositions could be applied to the transformations of nonlinear Gaussian innovations $\varepsilon$ as well as of nonlinear uniform innovations $u$. 

\subsubsection{Transformations in the Identified Set}

The definition of the identified set depends on the selected normalization of the nonlinear innovations either Gaussian, or uniform, say. In this section, we denote $u_t$, $v_t$, some potential uniform innovations and $\varepsilon_t$, $\eta_t$, potential Gaussian innovations, respectively. To make a link between these normalizations, we write:
\begin{equation}
	u_t= \begin{bmatrix}
		\Phi(\varepsilon_{1,t})\\
		\vdots \\
			\Phi(\varepsilon_{n,t})\\
	\end{bmatrix} \equiv \tilde{\Phi}(\varepsilon_t) \iff \varepsilon_t = \begin{bmatrix}
	\Phi^{-1}(u_{1,t})\\
	\vdots \\
	\Phi{-1}(u_{n,t})\\ 
\end{bmatrix}\equiv \tilde{\Phi}^{-1}(u_t),
\end{equation}
and similarily:
\begin{equation}\label{unif_renorm}
	v_t = \tilde{\Phi}(\eta_t) \iff \eta_t = \tilde{\Phi}^{-1}(v_t).
\end{equation}
To partly parameterize the identified set, it is easier to work with the Gaussian innovations and their ``polar" coordinates. Indeed, the support restrictions are much weaker than with the uniform distribution. More precisely, the vector $\varepsilon_t$ can be reparameterized as $||\varepsilon_t||^2 = \rho_t^2$, and $\frac{\varepsilon_t}{||\varepsilon_t||} = \zeta_t$. The first component $\rho^2_t$ is positive and the second component is an element of the $n-1$ hypersphere $S^{n-1}$. It is well-known that the standard Gaussian distribution for $\varepsilon_t$ becomes in polar coordinates $(\rho^2,\zeta_t)$ a continuous distribution with respect to the joint Lebesgue measure on $\mathbb{R}^+ \times$ the uniform probability measure on the sphere, with a joint density function of the type $\ell(\rho)$. Therefore, in polar coordinates any local transformation that keeps $\rho$ invariant and is orthogonal in $\zeta$ given $\rho$ will lead to the same distribution. This can be rewritten in terms of $\varepsilon_t$, $\eta_t$. 

\begin{proposition}
	Let $Q(\rho^2)$ be a continuous function of $\rho^2$ with values in the set of special orthogonal matrices [resp. $B(\rho^2)$ in the set of skew-symmetric matrices], then the nonlinear transformations $\eta_t = Q(||\varepsilon_t||^2) \varepsilon_t = \textup{Exp} B(||\varepsilon||^2)\varepsilon_t$, preserve the standard multivariate normal distribution.
\end{proposition}

These transformations are invertible with $\varepsilon_t = Q'(||\varepsilon_t||^2)\eta_t=\text{Exp} \left[-B(||\varepsilon||^2)\right] \eta_t$ and have their analogue in the uniform normalization using \eqref{unif_renorm}. 
\begin{corollary}
	The transformations $v_t=\tilde{\Phi}\left[Q(||\tilde{\Phi}^{-1}(u_t)||^2)\tilde{\Phi}^{-1}(u_t)\right]$, with $Q(\cdot)$ a special orthogonal function, preserve the uniform distribution on $[0,1]^n$. 
\end{corollary}

By construction, the transformations in Proposition 3 and Corollary 3 are nonlinear and their determinant is equal to 1. \\

The transformations above are easily understood in dimension 2, where it is easier to work in polar coordinates: $(\rho, \theta)$, instead of initial coordinates $(\varepsilon_{1,t},\varepsilon_{2,t})$. Then, the density of the standard normal distribution becomes $f(\rho,\theta)=\frac{1}{2\pi}\rho\exp(-\rho^2/2)$, on the domain $(0,\infty)\times(0,2\pi)$ (up to modulo $2\pi$ for $\theta$). Let us now consider the transformation $(\rho,\theta) \rightarrow (\rho^*,\theta^*) = (\rho,\theta+a(\rho))$, where $a(\cdot)$ is a continuous differentiable function of $\rho$. Since the Jacobian is $\left|\det\frac{\partial(\rho^*,\theta^*)}{\partial(\rho,\theta)'}\right|=\left|\det\begin{bmatrix}
	1 & 0 \\ \frac{\partial a(\rho)}{\partial \rho} & 1 \\ 
\end{bmatrix}\right|=1$, we see that the density of $(\rho^*,\theta^*)$ is also of the form $\frac{1}{2\pi}\rho^2 \exp(-\rho^2/2)$. For each fixed distance to the origin, the transformation is a rotation. When the rotation is with an angle $a$ independent of $\rho$, we get the standard linear orthogonal transform. But in general, we can make the rotation depending on the distance to the origin. In general we have [see Appendix B.3 for details]:
\begin{equation}
	\begin{cases}
		u^*_{1,t} = \Phi \left[\sqrt{\Phi^{-1}(u_{1,t})^2+\Phi^{-1}(u_{2,t})^2} \cos\left(\tan^{-1}(\Phi^{-1}(u_{2,t})/\Phi^{-1}(u_{1,t}))+a\left(\sqrt{\Phi^{-1}(u_{1,t})^2+\Phi^{-1}(u_{2,t})^2}\right)\right)\right] ,\\
		u^*_{2,t} = \Phi \left[\sqrt{\Phi^{-1}(u_{1,t})^2+\Phi^{-1}(u_{2,t})^2} \sin\left(\tan^{-1}(\Phi^{-1}(u_{2,t})/\Phi^{-1}(u_{1,t}))+a\left(\sqrt{\Phi^{-1}(u_{1,t})^2+\Phi^{-1}(u_{2,t})^2}\right)\right)\right] . \\ 
	\end{cases}
\end{equation}

 To illustrate these transformations and their effect on the uniform normalization, we consider the bidimensional case and the following transformations $\theta+a(\rho)$ on angle $\theta$:
\begin{enumerate}
	\item $a(\rho)=1$,
	\item $a(\rho) = 0.2\rho$,
%	\item $a(\rho)=\rho^2$,
%	\item $a(\rho)=\exp \rho$.
\end{enumerate}
For each case, we represent in $\begin{bmatrix}
	u_1 \\
	u_2 \\
\end{bmatrix} \in [0,1]^2$, the curves transformed of the segments $(u_1=\frac{k}{75}, 0\leq u_2 <1), k=0,1,...,75$ (resp. $0 \leq u_1 \leq 1, u_2 = \frac{k}{75}), k=0,1,...,75$. \\

\begin{figure}
	\centering
	\includegraphics[width=1\linewidth]{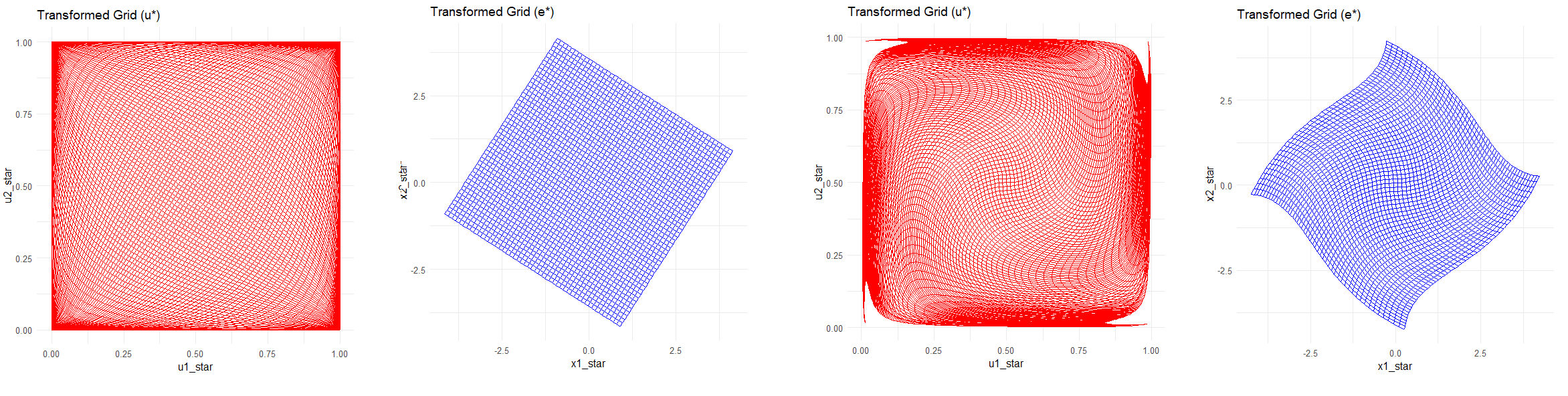}
	\caption{Piecewise Diffeomorphisms in $[0,1]^2$ in red and Diffeomorphisms in $\mathbb{R}^2$ in blue. In each set of plots, the first two correspond to $a(\rho)=1$, and the last two $a(\rho) = 0.2\rho$.}
	\label{fig:diffexamples}
\end{figure}

They correspond to homotopic deformations of the square [Smale (1959)]. These transformations can be equivalently represented for the Gaussian normalization. They lead to the corresponding  (truncated) diffeomorphisms in $\mathbb{R}^2$ which are plotted together in Figure 1 above.  \\

% given in Figure 2, where they are represented with the transformed lines $\varepsilon_2|\varepsilon_1=k$, \\

The set of transformations above shows the high degree of underidentification of the set of nonlinear innovations. Do alternative transformations exist, other than local orthogonal transforms that preserve the standard Gaussian? A result in this direction can be given in the bivariate case $n=2$, under the additional condition of zero preserving, known as the Schwarz Lemma, where the identified set is a set of functions. Without loss of generality, we can restrict the functions on $[0,1]^2$ to the harmonic functions that admit a series expansion as:
\begin{equation}
\begin{split}
		b_1(u)& = \sum_{h=0}^\infty \sum_{k=0}^\infty (b_{1,h,k}u_1^h u_2^k),\\
			b_2(u) &= \sum_{h=0}^\infty \sum_{k=0}^\infty (b_{2,h,k}u_1^h u_2^k),\\
\end{split}
\end{equation}
and analyze the identification issues by means of the coefficients $(b_{1,h,k}, b_{2,h,k}, h,k=0,1,...)$. We get the following result [see Gourieroux and Jasiak (2022)]:
\begin{proposition}
	In the bivariate case $n=2$, the set of identified harmonic functions is such that the manifold $\left\{b_{1,h,k}, b_{2,h,k},h+k\geq m\right\}$ is of dimension $2m$. 
\end{proposition}
\textbf{Proof:} See Appendix A.4. \\

Therefore, the degree of underidentification is rather large and the identified set contains transformations other than the nonlinear diffeomorphisms given as examples above.

\section{Identifying Restrictions and Identifiable Functions}

When the dimension $n$ is larger than 2, the nonlinear AR(1) process, that is function $g$ in \eqref{nlar}, is not identifiable. This is a problem of nonparametric identification with the identified set described in Proposition 2. As a consequence, the nonlinear innovations and the impulse response functions are not identifiable either. There are different ways to solve this identification issue and make shocks of interest more identifiable. [1] We can introduce parametric or nonparametric identifying restrictions making function $g$ identifiable. For instance, zero or sign restrictions can be considered to identify or partially identify a specific type of shock, such as a monetary policy shock [Uhlig (2005)], a fiscal policy shock [Ramey (2016), Section 4], a technology shock, or an oil shock [Hamilton (2003), Kilian and Murphy (2012)]. [2] We can use narrative records to identify shocks, as suggested in Romer and Romer (1997) [see also Antolin-Diaz and Rubio-Ramirez (2018)]. [3] Even if some partial identification is still present, we can look for transforms of the function $g$ (resp. the IRF) that are identifiable, even if they themselves are not. For instance, in the standard linear VAR framework, the horizonal Forecast Error Variance Decomposition (FEVD) is identifiable [see Gorodinchenko and Lee (2019) for the derivation of the FEVD by local projection and Gourieroux and Lee (2025) for extensions of the FEVD in the nonlinear dynamic framework]. We first briefly review the example of the Gaussian VAR(1) model, and then models with i.i.d. non-Gaussian error terms when the observations are preliminary conditionally demeaned and reduced. Then we explain how learning algorithms, simulations and pseudo IRFs can be used in an underidentified framework.

\subsection{Identifiable Function of the IRF}

Let us consider a SVAR(1) model and focus on the maximum possible IRF under a normalization on the magnitude of shock $\delta$. We consider the IRF on a linear combination of the variable $a'y$, say, that becomes $IRF(h,\delta,a)=a'A^hD\delta$, and on the solution of the constrained optimization problem:
\begin{equation}
	\max_\delta IRF(h,\delta,a), \ \text{s.t.} \ \delta'\delta=1.
\end{equation}
The Lagrangean $IRF(h,\delta,a)-\lambda \delta'\delta$ is optimized for:
\begin{equation*}
	\delta = \frac{1}{2\lambda} D'\Phi^{'h}a,
\end{equation*}
where $\lambda$ denotes the Langrange multiplier and the normalization condition implies:
\begin{equation*}
	2\lambda = \sqrt{a'\Phi^hDD'\Phi^{'h}a}.
\end{equation*}
We deduce the value of the objective function at its optimum. 
\begin{proposition}
	The maximum of $IRF(h,\delta,a), \ \text{s.t.} \ \delta'\delta=1$, is equal to $\sqrt{a'\Phi^hDD' \Phi^{'h}a}$, and is identifiable for any $h,a$. 
\end{proposition}
This approach can easily be extended to a specific type of shocks with sign restrictions on the IRF [see e.g. Plagborg-Moller and Wolf (2021), Example 3]. \\

An alternative is to introduce identifying restrictions linking the autoregressive parameter $A$ and the mixing matrix $D$. This is usually done on a ``structural form" of the model:
\begin{equation}
	\Phi_0y_t = \Phi_1y_{t-1} + \text{diag} \sigma \ \varepsilon_t,
\end{equation}
or equivalently: 
\begin{equation}
	y_t = \Phi^{-1}_0\Phi_1y_{t-1} + \Phi_0^{-1}\text{diag} \sigma \ \varepsilon_t,
\end{equation}
which is a VAR model with $A=\Phi^{-1}_0\Phi_1$, $D=\Phi_0^{-1}\text{diag}\sigma$. These can be exclusion (i.e. sparsity) restrictions with zero components in $\Phi_0$,$\Phi_1$, sign restrictions on IRF leading to partial identification only [Uhlig (2005), Antolin-Diaz and Rubio-Ramirez (2018)], or long run restrictions when technology shocks are defined as the only shocks that affect labour productivity in the long-run [Blanchard and Quah (1989)]. \\

\subsection{Conditionally Demeaned and Reduced Non Gaussian Models with Independence Restrictions}

They are written as:
\begin{equation}\label{cd_ar}
y_t = m(y_{t-1}) + D(y_{t-1})D^*w_t,
\end{equation}

where the components of $w_t$ are assumed to be independent. The specification \eqref{cd_ar} does not assume that the errors $w_t$ are either uniformly, or normally distributed. However, we can easily recover the nonlinear autoregressive specification with uniform innovations. More precisely, model \eqref{cd_ar} can be rewritten as:
\begin{equation}\label{uni_rep}
	y_t = m(y_{t-1}) + D(y_{t-1})D^*\begin{bmatrix}
		F_1^{-1}(u_{1,t})\\
		\vdots \\
			F_1^{-1}(u_{n,t})\\
	\end{bmatrix},
\end{equation} 
where $u_t$ follows $U[0,1]^n$ and $F_i$ denotes the c.d.f. of $w_{it}$, for $i=1,...,n$. The representation \eqref{uni_rep} shows the nonparametric restrictions introduced on the general autoregressive model \eqref{nlar}. The function $g(y_{t-1},u_t)$ is assumed linear affine in the specific transformations $F_i^{-1}$ of $u_{it}$, with coefficient functions of $y_{t-1}$. These restrictions allow for nonlinearities both in $y_{t-1}$ and $u_t$, but also on cross-effects. Nevertheless, these restrictions imply a significant reduction of the dimensions in terms of functions. In model \eqref{nlar}, function $g$ is from $\mathbb{R}^{2n}$ to $\mathbb{R}^n$ where $m(\cdot)$, $D(\cdot)$ depend on an argument of dimension $n$, and the $F_i^{-1}$'s depend on a one-dimensional argument. \\

% and the dimensions becomes $n^n + n^{n^2} + n^2$ (for $m(\cdot)$, $D(\cdot)$, and $F_i^{-1}$, respectively).\\

Model \eqref{cd_ar} includes the strong linear VAR(1) model where $m(y_{t-1})=Ay_{t-1}$ and $D(y_{t-1})=Id$. Let us assume that $m(y_{t-1})$ and $D(y_{t-1})$ are identified. Then we know asymptotically $ D^*w_t$. The problem of identification becomes a problem of multivariate deconvolution, also known as Independent Component Analysis (ICA). Can we identify $D^*$ and the distributions of $w_{1,t},...,w_{n,t}$ from the joint distribution of $D^*w_t$? The answer is given in the following proposition:
\begin{proposition}
	If there is at most one distribution of $u_{i,t}$ which is Gaussian, then we can identify $D^*$ and the sources $u_{1,t},...,u_{n,t}$ up to a scale effect on $u_{i,t}$ and permutation of the indices. We can also identify their distributions. 
 \end{proposition} 

\textbf{Proof:} See Comon (1994), Erikson and Koivunen (2004), Gouri\'eroux, Monfort and Renne (2017), Velasco (2022). \\

Thus, the linearity in $w_t$ in \eqref{cd_ar} and the independence assumptions make the model structural and facilitate the identification. In general the identified matrix $D^*$ is not triangular as assumed in the recursive definition of shocks by Sims (1980). Thus, the primitive shocks are not associated with an equation, but may enter in several equations, or in a combination of such equations. Proposition 6 shows also that, in such a model \eqref{cd_ar}, the degree of partial identification will depend on the number of Gaussian sources [See Guay (2021), Gouri\'eroux and Jasiak (2025)]. \\

\subsection{Identifiable Parameters}

As seen in Section 2, the distribution of the Markov process can be equivalently defined by means of its transition $\ell(y_{t}|y_{t-1})$, or by means of a nonlinear autoregressive representation. Usually the transition density is identifiable from a sequence of observations $y_t$, $t=1,...,T$, for instance if the process is stationary ergodic. Therefore, any function of this transition density is also identifiable. However, even under additional restrictions, such as a parametric form $\ell(y_t|y_{t-1};\theta)$, the nonlinear autoregressive model (that is $\theta$) is not always identifiable. Let us now discuss in detail the practice that consists of estimating nonparametrically the function $g$ (resp. the parameter $\theta$) by some learning algorithms, and the use of these estimates in the generative model \eqref{nlar} for simulations, predictions, or computations of the IRF where the identification issue still exists. 

\subsubsection{Learning Algorithms}

Learning algorithms to approximate the function $g$ (resp. the parameter $\theta$) are largely available. We focus on algorithms that require a starting value and then provide a well defined sequence of approximations $\hat{g}_{p,T}$ (resp. $\hat{\theta}_{p,T}$), where $p$ denotes the iteration step. They include gradient algorithms [Barzilai and Borwein (1988)], Average Stochastic Gradient Descent (ASGD) algorithms [Polyak and Juditsky (1992)], some Alternating Optimization algorithms and various algorithms used for ICA [Almeida (2003), Dinh, Krueger and Bengio (2015)]. Since we are in a framework, where $g$ (resp $\theta$) is not identifiable, we can just expect that some of these algorithms will provide for large $T$ and stopping rule $p_T$ an approximation close to the identified set. The approximation $\hat{g}_T$, say, depends in general on the selected starting value. It is out of the scope of this paper to derive theoretically the properties of these algorithms. We make the following assumptions:

\begin{enumerate}
	\item The learning algorithm provides for each starting value and each iteration step a unique value $\hat{g}_{p,T}$. In particular, some tuning parameters for accelerating the algorithm are supposed to be defined unambiguously. 
	\item $\hat{g}_{p_T,T}$ converges to the identified set when $T$, $p_T$ tend to infinity. 
\end{enumerate}
Can we use such an estimate in practice and for which purpose? 

\subsubsection{Simulations}

Let us first consider the simulation issues. Indeed, simulated trajectories are usually based on the nonlinear autoregressive specification in \eqref{nlar}. If $g_0$ and $g_1$ are two functions in the identified set, then we can simulate the series $y_{0t}^s$ and $y_{1t}^s$ from a same set of simulated errors $\varepsilon_t^s$, for $t=1,...,T$, with the same starting value $y_0$ by applying recursively:
\begin{equation}
	\begin{split}
	y_{0t}^s & = g_0(y_{0,t-1}^s,\varepsilon_t^s), \\
		y_{1t}^s & = g_1(y_{1,t-1}^s,\varepsilon_t^s). 
	\end{split}
\end{equation}
The two simulated trajectories will significantly differ since they are based on different transformations $g_0$ and $g_1$. However, since $g_0,g_1 \ \in \ IS$, their historical distributions will be the same for large $T$, and their simulated analogues based on a consistent $\hat{g}_T$, i.e.:
\begin{equation}
	\hat{y}_t^s = \hat{g}_T(\hat{y}_{t-1}^s,\varepsilon_t^s), \ t=1,...,T,
\end{equation}
will have the same asymptotic properties. To summarize, we can estimate consistently by simulations based on $\hat{g}_{p_T,T}$  any (functional) summary of the true distribution of the process, such as for instance the predictive distribution of $y_{t+h}$ given $y_t=y$. But the speed of convergence and asymptotic distribution of these estimators are not known in general.

\subsubsection{Pseudo IRF}

There exists a debate in Econometrics and Time Series on how to define shocks. Have they to be defined on an innovation $\varepsilon_t$, or directly on the series of interest $y_t$? [see e.g. Gallant et al. (1993) for this debate]. We have retained in Section 2 the definition based on innovations due to its more reasonable interpretation. Nevertheless, let us now change $y_t$ into $y_t + \Delta$, say, and let $\varepsilon_t$ be unchanged. The system \eqref{baseline} - \eqref{perturbed} becomes:
\begin{equation}
	y_{t+h} = g(y_{t+h-1},\varepsilon_{t+h}), \ h =0,1,...,
\end{equation}
\begin{equation}
	y^\Delta_{t+h} = g(y^\Delta_{t+h-1},\varepsilon_{t+h}), \ h =0,1,...,
\end{equation}
where $y_t^\Delta=y_t + \Delta$ and $y^\Delta_{t-1} = y_{t-1}$. This leads to the definition of the so-called Pseudo IRF (PIRF):
\begin{equation}
	PIRF_t(h,\Delta,y_t) = y_{t+h}^\Delta - y_{t+h}.
\end{equation}
Then we get the following result: 
\begin{proposition} The following results are true: 
	\begin{enumerate}
		\item $IRF_t(h,\delta,y_{t-1})=PIRF_t(h,g(y_{t-1},\varepsilon_t+\delta)-g(y_{t-1},\varepsilon_0),g(y_{t-1},\varepsilon_0))$.
		\item The function $PIRF_t$ is identifiable.
	\end{enumerate}
\end{proposition}
Therefore, the identification issue on the IRF is just due to the different effects at horizon 1 when applying a change on $y_t$ instead of a change on $\varepsilon_t$. The PIRF is in general stochastic and its expectation conditional on $y_t$ is equal to:
\begin{equation}
	\mathbb{E}\left[PIRF_t(h,\Delta,y_t)|y_t)\right]= \mathbb{E}(y_{t+h}|y_t+\Delta) - \mathbb{E}(y_{t+h}|y_t).
\end{equation}
This is a function of the transition density $\ell(y_t|y_{t-1})$ that can easily be consistently approximated by simulations as in Section 4.2.2. \\

\textbf{Remark 4:} The difference between the IRF and the Pseudo IRF is easily seen in the example of a structural Gaussian VAR model:
\begin{equation*}
	y_t = \Phi y_{t-1} + D \varepsilon_t,
\end{equation*}
where $\varepsilon_t \sim IIN(0,Id)$. We have:
\begin{equation*}
	\begin{split}
		IRF_t(h,\delta,y_{t-1}) & = (Id+\Phi + ... + \Phi^h) D\delta,\\
		PIRF_t(h,\Delta,y_t)& = (Id + \Phi + ... + \Phi^h) \Delta.
	\end{split}
\end{equation*}
The autoregressive coefficient $\Phi$ is identifiable, but not the rotation matrix $D$. This example is rather misleading since it gives the impression of similar term structures between the IRF and PIRF, independent of their environments. Proposition 6 shows that the link between them is much more complex in a nonlinear dynamic framework.

\section{Nonlinear Dynamic Factor Models}

The identification issues could also be solved when the multivariate observed time series is a transformation of independent dynamic factors (or sources in the operation research literature). This is the identification issue in dynamic nonlinear ICA, also called Blind Source Separation (BSS) to emphasize the difference with the i.i.d. static case [Jutten and Herault (1991), Hyvarinen and Morioka (2017)]. We consider below the bivariate case for expository purpose with linear and nonlinear transformations of the sources, respectively. 

\subsection{Linear Measurement Equation}
Let us assume that the bivariate observable process $(y_t)$ is such that:
\begin{equation}\label{linear_factor}
	y_t = A x_t,
\end{equation}
where the sources $(x_{1,t})$ and $(x_{2,t})$ are independent Markov processes:
\begin{equation}
	\begin{cases}
		x_{1,t} = g_1(x_{1,t-1};\varepsilon_{1,t}), \\
				x_{2,t} = g_2(x_{2,t-1};\varepsilon_{2,t}), \\
	\end{cases}
\end{equation}
$(\varepsilon_{1,t})$, $(\varepsilon_{2,t})$ are independent Gaussian noises, and matrix $A$ is invertible. The nonparametric identification of $A$, $g_1$, $g_2$, is equivalent to the nonparametric identification of matrix $A$. Note that we get the same type of dynamic model if $(x_{1,t})$ (or $(x_{2,t})$) is multiplied by a (signed) scale factor, and also if we permute the one-dimensional processes $(x_{1,t})$ and $(x_{2,t})$. This leads to the following definition: \\

\begin{definition}
	The representation $(A,g_1,g_2)$ with linear transformation $A$ is essentially unique if and only if it is identified up to a signed scale effect and permutation of indexes. 
\end{definition}

The lack of identification due to scale effects can be solved by introducing an appropriate identification restriction as: ``the diagonal elements of matrix $A$ are equal to 1". Then, we have the following result:

\begin{proposition}
	Let us denote $\tilde{\gamma}_j(h), j=1,2$, the autocovariance functions of processes $(x_{1,t})$ and $(x_{2,t})$. Then, if ($\tilde{\gamma}_1(h)$, $h$ varying) and ($\tilde{\gamma}_2(h)$, $h$ varying) are linearly independent, the representation is essentially unique. 
\end{proposition}

\textbf{Proof:} See Appendix A.5. \\ 

We have seen in Proposition 6 that there is an identification issue when $(x_{1,t})$ and $(x_{2,t})$ are independent Gaussian white noises. This lack of identification disappears, if the underlying independent factors (sources) have different second-order dynamics. This result can be extended to any dimension where the observations are combinations of independent AR(1) processes with different autoregressive coefficients in [Gouri\'eroux and Jasiak (2023b)]. Autocovariance based methods to separate sources in the linear framework have been proposed in Tong et al. (1991) and Belouchani et al. (1997). The condition in Proposition 8 is sufficient, but not necessary. Weaker identification conditions can sometimes be obtained by considering zero autocovariance restrictions or nonlinear transformations of $x_{1,t}$ (resp. $x_{2,t}$) [Gouri\'eroux and Jasiak (2022b)]. The idea is the following. Let us assume that the dynamics of the latent factors are parameterized, that is:
\begin{equation*}
\begin{cases}
		x_{1,t}  = g_1(x_{1,t-1},\varepsilon_{1,t};\theta_1),\\
	x_{2,t}  = g_2(x_{2,t-1},\varepsilon_{2,t};\theta_2),\\
\end{cases}
\end{equation*}
where $\theta_1$, $\theta_2$ are parameters. This system can be inverted as:
\begin{equation*}
\begin{cases}
		\varepsilon_{1,t}  = h_1(x_{1,t},x_{1,t-1};\theta_1),\\
	\varepsilon_{2,t}  = h_2(x_{2,t},x_{2,t-1};\theta_2).\\ 
\end{cases} \iff \varepsilon_{t} = h(x_t,x_{t-1};\theta).
\end{equation*}
Let us denote $C=A^{-1}$ as the demixing matrix. Then, the model can be rewritten as:
\begin{equation*}
	\varepsilon_t = h^*(Cy_t,Cy_{t-1};\theta^*),
\end{equation*}
with $\theta=(\text{vec} C,\theta^*)$. Since the $\varepsilon_t's$ are serially independent, it has been suggested in the literature to try to identify the parameters $C$ and $\theta$ by means of a set of nonlinear autocovariance restrictions of the form:
\begin{equation}
	Cov\left[a(\varepsilon_t),\tilde{a}(\varepsilon_{t-k})\right] = Cov\left[a \circ h^*\left[Cy_{t},Cy_{t-1};\theta^*\right],\tilde{a} \circ h^*\left[Cy_{t-k},Cy_{t-k-1};\theta^*\right]\right]=0,
\end{equation}
with a selected set of pairs of nonlinear transformations $a,\tilde{a}$ and lags. Such an identification approach is for instance introduced for the identification of mixed causal-noncausal models when using the Generalized Covariance (GCov) estimation approach. 

\subsection{Nonlinear Measurement Equation} 

The model in the previous subsection can be extended to account for nonlinear transformations of the sources, where \eqref{linear_factor} is replaced by:
\begin{equation}
	y_t = A(x_t).
\end{equation}
Then, the definition of essential identification has to be extended to account for the nonlinearity of the transformation $A$. Typically, $A$, $g_1$ and $g_2$ can be modified for the latent variables $(x_{1,t})$ and $(x_{2,t})$ to be marginally uniform on $[0,1]$, for $(y_t)$ to be marginally uniform on $[0,1]^2$, and up to permutation of indices. \\

\begin{definition}
	The representation $(A,g_1,g_2)$ with nonlinear transformation $A$ is essentially unique if and only if it is unique up to a transformation on $(x_t)$ and $(y_t)$ to make them marginally uniform and up to permutation  of indices. 
\end{definition}
By fixing the margins, the analysis of identification will be through pairwise copulas at different lags instead of covariances at different lags as in Proposition 8. We have:   

\begin{proposition}
	Let us denote $\ell_{jh}(x_{j,t+h}|x_{j,t})$, $j=1,2,$ as the conditional transitions at horizon $h$ of the two sources. Let us assume that the functions $\ell_{jh}(u|v)$ are differentiable  with respect to $u$ and that the sequences $\left[\frac{\partial \log \ell_{jh} }{\partial u}(u|v), \ \text{v varying}\right]$, $j=1,2$, are linearly independent. Then $A,g_1,g_2$ are essentially unique. 
\end{proposition}

\textbf{Proof:} See Appendix A.6. \\

Therefore, there is no identification issue if the sources are independent with different dynamics, even in the Gaussian framework. As an illustration, let us consider Gaussian AR(1) sources: $x_{j,t}=\rho_jx_{j,t-1}+\sigma_j\varepsilon_{j,t}$. The conditional transition at horizon $h$ is Gaussian $N(\rho^hx_{j,t},\eta^2_{j,h})$, where $\eta^2_{j,h}=\sigma_j^2\frac{1-\rho^{2h}_j}{1-\rho^2_j}$. Then we have:
\begin{equation*}
	\frac{\partial\log \ell_{jh}}{\partial u} (u|v) = - \frac{u}{\eta^2_{j,h}} + \frac{\rho^h_jv}{\eta^2_{j,h}}, j=1,2,
\end{equation*}
and the condition of Proposition 8 is satisfied if and only if $\rho_1 \neq \rho_2$. 

\subsection{Impulse Response Functions}

Under the condition of Proposition 8, $\varepsilon_t = \begin{bmatrix}
	\varepsilon_{1,t} \\
	\varepsilon_{2,t}
\end{bmatrix}$ is a Gaussian nonlinear innovation defined in a unique way. Therefore, the associated IRF for $y_t$ are also identifiable. They are easily deduced from the IRF performed separately on $x_{1,t}$ and $x_{2,t}$. More precisely we define the shocked and unshocked trajectories of the sources with respective shocks $\delta_1$ and $\delta_2$. These are $x_{1,t+h}$, $x_{2,t+h}$, $x^{(\delta_1)}_{1,t+h}$ and  $x^{(\delta_2)}_{2,t+h}$. Then the shocked and unshocked trajectories of $y_t$ are: 
\begin{equation*}
	\begin{split}
		y_{t+h}^{(\delta)} = A\begin{bmatrix}
			x^{(\delta_1)}_{1,t+h}\\x^{(\delta_2)}_{2,t+h}
		\end{bmatrix} \ \ \text{and} \ \ 	y_{t+h} = A\begin{bmatrix}
		x_{1,t+h}\\x_{2,t+h}
	\end{bmatrix}.
	\end{split}
\end{equation*}
Due to the nonlinear transformation $A$, the $IRF(h,\delta,y_{t-1})$ is not a function of the IRF's of the two sources only. In other words, in a nonlinear dynamic framework, the notion of IRF is not coherent with respect to nonlinear aggregation. This approach is using implicitly the assumption of equal dimensions of the processes $(y_t)$ and $(x_t)$. If $(x_t)$ were with a dimension strictly larger than the dimension of $(y_t)$, the information in $\underline{y_t}$ and $\underline{x_t}$ would not be the same. 

\section{Augmented Nonlinear Autoregressive Models}

It is known that the notion of causal inference, of Granger causality, and likely also the notion of Markov process, nonlinear innovation and IRF can depend on the universe, that is on the selected set of variables. For instance, narrative shocks are often found to be predictable, suggesting the possibility of endogeneity [Leeper (1997), Romer and Romer (1997); see also Ganies et al. (2019) for an IV adjustment for endogeneity]. In particular, starting from a given universe, they could be modified if the universe is either decreased, or increased. This leads for instance to the introduction of the Factor Augmented Vector Autoregressive (FAVAR) models in the linear dynamic VAR literature [Bernanke, Boivin and Eliasz (2002)]\footnote{These augmented vector autoregressive models concern the true dynamic of the variable. They have to be distinguished from the lag augmented local projections, that are instrumental models introduced for robust inference on standard errors of the IRF [Montiel-Olea and Plagborg-Moller (2021), (2022)]}. 

\subsection{The Markov Assumption}

In fact there is a trade-off between the dimension of the system and the order of the Markov process\footnote{Assumed equal to 1 in our paper, but the discussion is the same for any fixed order $p$ [see e.g. Florens et al. (1993)].}. For instance, it is well known that a one-dimensional AR(2) process can be written as a bidimensional VAR(1) process. It is also known that the first component of a bivariate VAR(1) process satisfies an ARMA model (i.e. an AR($\infty$) model), when its dynamic is marginalized. This highlights the advantage of multivariate modelling to avoid the increased autoregressive order. \\

To understand why the Markov assumption depends on the universe, let us consider a (multivariate) process $(y_t)$, and the process augmented by another (multivariate) process ($\zeta_t$). We denote $\underline{y}_t$ (resp. $\underline{\zeta}_t$) the set of current and past values of the process $(y_t)$(resp. $(\zeta_t)$). There are three notions of Markov processes with conditions written in terms of transition densities. \\

(i) The augmented process $(y_t,\zeta_t)$ is Markov with respect to the augmented universe: \\
	\begin{equation}\label{markov_1} 
		\ell(y_t,\zeta_t| \ \underline{y}_{t-1},\underline{\zeta}_{t-1}) = \ell(y_t,\zeta_t| \ y_{t-1},\zeta_{t-1}). 
	\end{equation}

(ii) The initial process $(y_t)$ is Markov with respect to its own universe: \\
\begin{equation}\label{markov_2} 
	\ell(y_t| \ \underline{y}_{t-1}) = \ell(y_t| \ y_{t-1}).
\end{equation}

(iii) The initial process $(y_t)$ is Markov with respect to the augmented universe: \\
\begin{equation}\label{markov_3} 
	\ell(y_t| \ \underline{y}_{t-1},\underline{\zeta}_{t-1})=\ell(y_t| \ y_{t-1}). 
\end{equation}

All these conditions can be rewritten in terms of conditional independence, denoted $\boldsymbol{\cdot} \indep \boldsymbol{\cdot} \ |  \ \boldsymbol{\cdot}$ [Florens et al. (1993)] as: 

(i) $(y_t,\zeta_t)$ is Markov in the augmented universe, iff $(y_t,\zeta_t) \indep (\underline{y}_{t-2},\underline{\zeta}_{t-2})  | \  y_{t-1},\zeta_{t-1}$. \\
(ii) $(y_t)$ is Markov in its universe, iff $y_t \indep \underline{y}_{t-2} | \  y_{t-1}$. \\
(iii) $(y_t)$ is Markov in its augmented universe, iff $y_t \indep (\underline{y}_{t-2},\underline{\zeta}_{t-1})  |  \ y_{t-1}$. \\

It is easily seen that, under \eqref{markov_3}, the condition \eqref{markov_2} means that process $(\zeta_t)$ does not (Granger) cause process $(y_t)$, that is a condition of dynamic exogeneity of process $(y_t)$. \\

In the previous sections, to define the nonlinear innovations, to define the IRF corresponding to the shocks on their innovation and to discuss the identification issues, we have assumed that the nonlinear VAR(1) model is well-specified, that is the Markov assumption is satisfied. This predictive assumption has now to be discussed. In particular, we have to consider the following question - in which case do we have simultaneously the joint and marginal Markov properties \eqref{markov_1},\eqref{markov_2} satisfied?\\ 

We have the following result:

\begin{proposition}
	Let us assume that the augmented process $(y_t,\zeta_t)$ is Markov of order 1 and that the process $(\zeta_t)$ does not cause $(y_t)$. Then, $(y_t)$ is Markov with respect to its own universe. 
\end{proposition}

\begin{proof}
	Indeed, we have:
	\begin{equation*}
		\begin{split}
			\ell(y_t| \ \underline{y}_{t-1},\underline{\zeta}_{t-1}) & = \ell(y_t | \ y_{t-1},\zeta_{t-1}) = \ell(y_t|y_{t-1})\\ 
	%		& = \ell(y_t|y_{t-1}) . 
		\end{split}
	\end{equation*}
Then, reintegrating with respect to $\underline{\zeta}_{t-1}$ both sides, we get:
\begin{equation*}
	\ell(y_t| \ \underline{y}_{t-1}) = \ell(y_t | \ y_{t-1}).
\end{equation*}
that is condition \eqref{markov_2}. 
\end{proof}

This noncausality condition is sufficient, but not necessary [see Florens et al. (1993) for more detailed analysis]. This type of condition will also appear when the IRFs are obtained by ``state dependent" local projection [Goncalves et al. (2022), (2023)]. 

\subsubsection{Testing the Markov Assumption}

The discussion in Section 5.3.1 shows that the Markov assumption is at the core of the definitions of IRFs and their identification. Therefore, in practice, this assumption would have to be tested (which is rarely the case, even in linear dynamic VARs). In the identified case, it is equivalent to test that the nonlinear innovations are independent white noises. This can be done by applying tests of independence and of strong white noise to the series of nonlinear residuals. These tests have to check for the absence of any type of cross-sectional and serial nonlinear correlation (not only linear ones). When there is an identification issue, tests of conditional independence can be used based on the interpretation of the Markov assumption in terms of conditional independence. This condition can be written in terms of appropriate moments. For instance, the Markov condition [ii] is equivalent to the covariance conditions:
\begin{equation}
	\begin{split}
	&\mathbb{E}[(a(y_t)b(y_{t-2})c(y_{t-1})] = \mathbb{E}\{\mathbb{E}[a(y_t)]\mathbb{E}[b(y_{t-2})]c(y_{t-1})\} \\ 
	\iff&\mathbb{E}[Cov(a(y_t),b(y_{t-2}))c(y_{t-1})]=0,
	\end{split}
\end{equation}
for any integrable functions $a,b,c$. Such conditions can be checked by using pormanteau test statistics based on the sample covariance counterparts [Gourieroux and Jasiak (2023a),(2025), Velasco (2023)].

\section{Concluding Remarks}

This paper has extended the notions of nonlinear autoregressive representation, nonlinear innovations and impulse response functions to the multivariate nonlinear dynamic framework.  We analyzed the identification issue for these innovations and IRFs. In this respect, the linear Gaussian VAR models and the unconstrained Markov model appear as special cases in which the identification issues are especially significant. The main identification issues can disappear when nonlinearities and/or independent latent sources with different dynamics exist. Such analysis can be applied to any nonlinear dynamic system. Can we use them as tools for economic policy? Following Bernanke (1986), p. 52-55, ``Structures of shocks are primitive exogenous forces (i.e. the innovations), that are independent of each other and economically meaningful". Therefore, in practice, the estimation of the model, of the innovation, of the IRFs have to be followed by a structural step to interpret economically some of these innovations, not necessarily all of them. All the analyses in this paper have been developed for time series with multivariate observations indexed by time $t$. It would be interesting to extend this analysis to series that are doubly indexed: [1] By time and maturity when considering the term structure of interest rates or of volatilities, [2] indexed by the localization for spatial processes, or [3] indexed by a pair of individuals as in networks. These extensions are left for further research.  

\newpage
	
\section{References}

\setstretch{1}

Almeida, L. (2003). MISEP—Linear and Nonlinear ICA Based on Mutual Information. The Journal of Machine Learning Research, 4, 1297-1318.\\

Antolín-Díaz, J., and  J., Rubio-Ramírez. (2018). Narrative Sign Restrictions for SVARs. \textit{American Economic Review}, 108, 2802-2829.\\

Axler, S., Bourdon, P. and W., Ramey (2001): ``Harmonic Function Theory", 2nd Ed., Graduate
Texts in Mathematics, Springer.\\

Barzilai, J., and M., Borwein (1988). Two-Point Step Size Gradient Methods. IMA Journal of Mathematical Analysis, 8, 141-148. \\

Belouchani, A., Meraim, K., Cardoso, S., and E., Moulines (1997): A Blind Source Separation Technique Based on Second-Order Statistics", IEEE Trans on Signal Processing, 45, 434-444. \\

Bernanke, B. (1986). Alternative Explanations of the Money-Income Correlation. \textit{Carnegie Rochester Conference Series on Public Policy}. 25, 49-99. \\

Bernanke, B., Boivin, J., and P., Eliasz. (2002). Measuring the Effects of Monetary Policy: A Factor Augmented Vector Autoregressive (FAVAR) Approach. \textit{Quarterly Journal of Economics}, 120, 387-422. \\

Blanchard, O., and D., Quah. (1989). The Dynamic Effects of Aggregate Demand and Supply Disturbances. \textit{American Economic Review}, 73, 655-673. \\ 

Borkovec, M., and C., Kluppelberg. (2001). The Tail of the Stationary Distribution of an Autoregressive Process with ARCH(1) Errors. \textit{Annals of Applied Probability},  11, 1220-1241. \\

Christiano, L. (2012). Christopher A. Sims and Vector Autoregressions. \textit{The Scandinavian Journal of Economics}, 114, 1082-1104. \\ 

Comon, P. (1994). Independent Component Analysis, a New Concept?. \textit{Signal Processing}, 36, 287-314. \\ 

Cox, J., Ingersoll, J., and S., Ross. (2005). A Theory of the Term Structure of Interest Rates. \textit{Theory of Valuation},129-164. \\ 

Darmois, G. (1953). Analyse Generale de Liaisons Stochastiques. Rev Inst Internat Stat, 21, 2-8. \\

Dinh, L., Krueger, D., and Y., Bengio (2015): ``NICE: Nonlinear Independent Components Estimation", ICLR. \\

Eriksson, J., and V., Koivunen. (2004). Identifiability, Separability, and Uniqueness of Linear ICA models. \textit{IEEE Signal Processing Letters}, 11, 601-604. \\

Florens, J.P., Mouchart, M. and J.M., Rolin. (1993). Noncausality and Marginalization of Markov Processes. \textit{Econometric Theory}, 9, 239-260. \\

Gallant, A., Rossi, P., and  G., Tauchen. (1993). Nonlinear Dynamic Structures. \textit{Econometrica}, 61, 871-907.\\

Ganics, G., Inoue, A., and B., Rossi. (2021). Confidence Intervals for Bias and Size Distortion in IV and Local Projections-IV Models. \textit{Journal of Business \& Economic Statistics}, 39(1), 307-324. \\ 

Goncalves, S., Herrera, A., Kilian, L., and E., Pesavento. (2021). Impulse Response Analysis for Structural Dynamic Models with Nonlinear Regressors. \textit{Journal of Econometrics}, 225, 107-130. \\ 

Goncalves, S., Herrera, A., Kilian, L., and E., Pesavento. (2022). When Do State Dependent Local Projections Work?, DP McGill University. \\

Gorodnichenko, Y., and B., Lee. (2020). Forecast Error Variance Decompositions with Local Projections. \textit{Journal of Business and Economic Statistics}, 38(4), 921-933.\\

Gouri\'eroux, C.,  and J., Jasiak. (2005). Nonlinear Innovations and Impulse Responses with Application to VaR Sensitivity. \textit{Annales d'Economie et de Statistique}, 78, 1-31. \\

Gouri\'eroux, C.,  and J., Jasiak. (2022). Nonlinear Forecasts and Impulse Responses for Causal-Noncausal (S)VAR Models. \textit{arXiv preprint arXiv:2205.09922}. \\ 

Gouri\'eroux, C.,  and J., Jasiak. (2023a). Generalized Covariance Estimator. \textit{Journal of Business and Economic Statistics}, 41, 1315-1327. \\ 

Gouri\'eroux, C., and J., Jasiak. (2023b). Dynamic Deconvolution of Independent Autoregressive Sources. \textit{Journal of Time Series Analysis}, 44, 151-180. \\ 

Gouri\'eroux, C., and J. Jasiak. (2025). Generalized Covariance Based Inference for Models Partially Identified from Independence Restrictions, \textit{Journal of Time Series Analysis}, 46, 300-324. \\

Gouriéroux, C., Jasiak, J., and R., Sufana. (2009). The Wishart Autoregressive Process of Multivariate Stochastic Volatility. Journal of Econometrics, 150(2), 167-181.\\

Gouri\'eroux, C., and Q., Lee. (2025). Forecast Relative Error Decompositions with Application to Cyber Risk. arXiv:2406.17708.\\

Gouri\'eroux, C., and A., Monfort. (1997). Simulation Based Econometric Methods, Oxford University Press. \\

Gouri\'eroux, C., Monfort, A.,  and J.P., Renne. (2017). Statistical Inference for Independent Component Analysis: Application to Structural VAR Models. \textit{Journal of Econometrics}, 196, 111-126. \\ 

Guay, A. (2021). Identification of Structural Vector Autoregressions Through Higher Unconditional Moments. \textit{Journal of Econometrics}, 225, 27-46.\\

Hall, B. (2015). Lie Groups, Lie Algebras and Representations: An Elementary Introduction", Graduate Texts in Mathematics, 222, (2nd Edition), Springer. \\

Hamilton, J. (2003). What is an Oil Shock?. \textit{Journal of Econometrics}, 113, 363-398. \\

Hyvarinen, A., and H., Morioka. (2017). Nonlinear ICA of Temporally Dependent Stationary Sources. iin Proceedings of the 20th International Conference on Artificial Intelligence and Statistics, 54, 460-469. \\ 

Hyvarinen, A., and P., Pajunen. (1999). Nonlinear Independent Component Analysis: Existence and Uniqueness Results. \textit{Neural Networks}, 12, 429-439.\\ 

Hyvärinen, A., and E., Oja (2001). Independent Component Analysis: Algorithms and Applications. Neural Networks, 13(4-5), 411-430.\\

Hyvarinen, A., Sasaki, H., and R., Turner (2019). Nonlinear ICA Using Auxiliary Variables and Generalized Contrastive Learning. in Proceedings of the 22nd International Conference on Artificial Intelligence and Statistics, 89, 859 -868. \\ 

Isakin, M.,  and P., Ngo. (2020). Variance Decomposition Analysis for Nonlinear Economic Models. \textit{Oxford Bulletin of Economics and Statistics}, 82, 1362-1374. \\

Jutten, C., and J., Herault. (1991). Blind Separation of Sources Part 1: An Adaptive Algorithm Based on Neuromimetic Architecture. Signal Processing. 94. 1-10. \\

Karlin, S., and H., Taylor. (1981). A Second Course in Stochastic Processes. \textit{Elsevier}. \\ 

Karlsen, M. and D., Tjostheim. (2001). Nonparametric Estimation on Null Recurrent Time Series. \textit{Annals of Statistics}, 29, 372-416. \\

Kilian, L., and D., Murphy. (2012). Why Agnostic Sign Restrictions are Not Enough: Understanding the Dynamics of Oil Market VAR Models. \textit{Journal of the European Economic Association}, 10, 1166-1188.\\

Koop, G., Pesaran, H.,  and S., Potter. (1996). Impulse Response Analysis in Nonlinear Multivariate Models. \textit{Journal of Econometrics}, 74, 119-147. \\

Leeper, E. (1997). Narrative and VAR Approaches to Monetary Policy: Common Identification Problems. \textit{Journal of Monetary Economics}, 40, 641-657. \\

Ling, S. (2007). A Double AR(p) Model: Structure and Estimation. \textit{Statistica Sinica}, 17, 161-175. \\

Montiel Olea, J., and M., Plagborg‐Møller. (2021). Local Projection Inference is Simpler and More Robust Than You Think. \textit{Econometrica}, 89, 1789-1823. \\ 

Montiel Olea, J., and M., Plagborg‐Møller. (2022). Corrigendum: Local Projection Inference is Simpler and More Robust Than You Think. Online Manuscript. \\ 

Nakamura, E., and J., Steinsson. (2018). Identification in Macroeconomics. \textit{Journal of Economic Perspectives}, 32, 59-86.\\

Plagborg‐Møller, M., and C., Wolf. (2021). Local Projections and VARs Estimate the Same Impulse Responses. \textit{Econometrica}, 89, 955-980. \\

Polyak, B., and A., Juditsky (1992). Acceleration of Stochastic Approximation by Averaging. SIAM Journal on Control and Optimization, 30(4), 838-855.\\

Ramey, V. (2016). Macroeconomic Shocks and their Propagation. \textit{Handbook of Macroeconomics}, 2, 71-162. \\ 

Roberts, S., and R., Everson (2001). Independent Component Analysis: Principles and Practice. Cambridge University Press. \\

Romer, C.,  and  D., Romer. (1994). Monetary Policy Matters. \textit{Journal of Monetary Economics}, 34, 75-88.\\ 

Romer, C.,  and  D., Romer. (1997). Identification and the Narrative Approach: A Reply to Leeper. \textit{Journal of Monetary Economics}, 40, 659-665.\\ 

Rosenblatt, M. (1952). Remarks on a Multivariate Transformation. The Annals of Mathematical Statistics, 23, 470-472.\\

Rubio-Ramirez, J. , Waggoner, D., and T., Zha. (2010). Structural Vector Autoregressions: Theory of Identification and Algorithms for Inference. \textit{The Review of Economic Studies}, 77, 665-696.\\

Sims, C. (1980). Macroeconomics and Reality. \textit{Econometrica}, 48, 1-48. \\ 

Smale, S. (1959). Diffeomorphisms of the 2-Sphere. \textit{Proceedings of the American Mathematical Society}, 10, 621-626. \\

Tong, L., Liu, R., Soon, V., and Y., Huang. (1991). Indeterminacy and Identifiability of Blind Identification. IEEE Transactions on Circuits and Systems, 38(5), 499-509.\\

Tweedie, R. (1975). Sufficient Conditions for Ergodicity and Recurrence of Markov Chains on a General State Space. \textit{Stochastic Processes and their Applications}, 3, 385-403. \\ 

Uhlig, H. (2005). What are the Effects of Monetary Policy on Output? Results from an Agnostic Identification Procedure. \textit{Journal of Monetary Economics}, 52, 381-419.\\

Velasco, C. (2023). Identification and Estimation of Structural VARMA Models Using Higher Order Dynamics. \textit{Journal of Business \& Economic Statistics}, 41, 815-832.\\

Weiss, A. (1984). ARMA Models with ARCH Errors. \textit{Journal of Time Series Analysis}, 5, 129-143. \\

Zhu, H., Zhang, X., Liang, X., and Y., Li. (2017). On a Vector Double Autoregressive Model. \textit{Statistics \& Probability Letters}, 129, 86-95.\\

\newpage
	
\appendix

\begin{appendices}
	\setstretch{1}
%Note: All Appendices to be put online.

\section{Identification of Sources}
\setstretch{1}

\subsection{Proof of Proposition 1}
The proof extends the inversion method given in \eqref{fn_g} to the multivariate case [see e.g. Rosenblatt (1952), Gouri\'eroux and Monfort (1997), Hyvarinen and Pajunen (1999)]. Let us denote $\underline{y_{i-1,t}}=(y_{1,t},...,y_{i-1,t})$ and define: 
\begin{equation*}
	F_{i,t}(y|\underline{y_{i-1,t}},\underline{y_{t-1}})=\mathbb{P}[y_{i,t}<y|\underline{y_{i-1,t}},\underline{y_{t-1}}], i=1,...,n.
\end{equation*}
Then we can define: 
\begin{equation*}
	\begin{split}
		\varepsilon_{i,t} & = \Phi^{-1} Q_{i,t}[y_{i,t}|\underline{y_{i-1,t}},\underline{y_{t-1}}] \\
		& = \Phi^{-1}F^{-1}_{i,t}[y_{i,t}|\underline{y_{i-1,t}},\underline{y_{t-1}}], i = 1,...,n, 
	\end{split}
\end{equation*}
where  $F_{i,t}$, $Q_{i,t}$ denote the appropriate conditional c.d.f. and quantile functions for $y_{i,t}$ and $\Phi$ is the c.d.f of the standard normal distribution. The conditional distribution of $\varepsilon_{i,t}$ given $\underline{y_{i-1,t}},\underline{y_{t-1}}$ is N(0,1), independent of the conditioning set. Moreover, $\varepsilon_{1,t},...,\varepsilon_{i-1,t}$ are functions of $\underline{y_{i-1,t}},\underline{y_{t-1}}$. Therefore $\varepsilon_{i,t}$ is also independent of $\underline{\varepsilon_{i-1,t}}=(\varepsilon_{1,t},...,\varepsilon_{i-1,t})$ and $\underline{y_{t-1}}$. We deduce that, conditional on $\underline{y_{t-1}}$, the vector $\varepsilon_t$ is multivariate Gaussian $N(0,Id)$. QED.

\subsection{Irregular Transformations Keeping Invariant N(0,1)}

The multiplicity of nonlinear innovations is even larger if we do not impose the differentiability and/or monotonicity of the transformation. To illustrate this point, let us consider the one-dimensional framework, assume $\varepsilon\sim N(0,1)$ and define:
\begin{equation*}
	y = \begin{cases}
		\varepsilon, \ \text{if} \ |\varepsilon|<c, \\ 
		- \varepsilon, \ \text{if} \ |\varepsilon|>c,
	\end{cases}
\end{equation*}
where $c$ is a given positive threshold. Then, by the symmetry of the normal distribution, we see that $y$ also follows a standard normal $N(0,1)$.

\subsection{Proof of Proposition 2}

The nonlinear autoregressive process: $y_t=g(y_{t-1};\varepsilon_t)$ can be equivalently written as $y_t=h(y_{t-1};u_t)$, where the components of $u_t$ are i.i.d. uniform on [0,1] and $h$ is another continuously differentiable function w.r.t. $u_t$, with continuously differentiable inverse. Let us now consider another nonlinear autoregressive representation:
\begin{equation*}
	y_t = h(y_{t-1};u_t) = \tilde{h}(y_{t-1};\tilde{u}_t).
\end{equation*}
Then we have: $\tilde{u}_t=\tilde{h}^{-1}(y_{t-1};h(y_{t-1};u_t))=T(u_t)$, say. Indeed $\tilde{u}_t$ cannot depend on $y_{t-1}$ by the independence between $\tilde{u}_t$ and $y_{t-1}$. Then necessarily $\tilde{u}_t$ is a continuously differentiable function of $u_t$ with continuously differentiable inverse. Let us now derive the distribution of $\tilde{u}_t$. By the Jacobian formula, this distribution is continuous with density $\tilde{f}$ on $[0,1]^n$ such that: 
\begin{equation}
	f(u) = \left|\det\frac{\partial T(u)}{\partial u'}\right|\tilde{f}[T(u)], \ \text{on} \ [0,1]^n,
\end{equation}
where $f$ is the density of $u$. Since both $f$ and $\tilde{f}$ are uniform densities on $[0,1]^n$, we deduce:
\begin{equation*}
	\left|\det\frac{\partial T(u)}{\partial u'}\right| = 1, \ \forall u \ \text{in} \ [0,1]^n.
\end{equation*}

Moreover, by the global invertibility of the transformation $T$, this condition implies either $\det\frac{\partial T(u)}{\partial u'}=1, \ \forall u \ \text{in} \ [0,1]^n$, or $\det\frac{\partial T(u)}{\partial u'}=-1, \ \forall u \ \text{in} \ [0,1]^n$. The result follows. \qed

\subsection{Proof of Proposition 4}

\textbf{a) Existence of nonlinear causal autoregressive representation}\\

For ease of exposition, let us consider a bivariate  Markov process. By analogy to the recursive  causal approach for defining the shocks, we start from the first component.

i) Let $F_1[y_1 | Y_{T-1}]$ denote the conditional c.d.f. of $Y_{1,T}$ given $Y_{T-1}$ and define:

\begin{equation}
v_{1,T} = F_1[Y_{1,T}|Y_{T-1}], \; \forall T.
\end{equation}

 Then, $v_{1,T}$ follows a uniform distribution $U_{[0,1]}$ for any $Y_{T-1}$. In particular, $v_{1,T}$ is independent of $Y_{T-1}$.\\
 
ii) Let $F_2[y_2 | Y_{1,T}, Y_{T-1}]$ denote the conditional c.d.f. of $Y_{2,T}$ given
$ Y_{1,T}, Y_{T-1}$, and define: 
\begin{equation}
v_{2,T} = F_2[Y_{2,T} | Y_{1,T}, Y_{T-1}], \; \forall T.
\end{equation}
It follows that $ v_{2,T}$ follows a uniform distribution on [0,1], for any $Y_{1,T}, Y_{T-1}$,
or equivalently for any $v_{1,T}, Y_{T-1}$. Therefore, $ v_{2,T}$ is independent of $v_{1,T}, Y_{T-1}$.\\

iii) By inverting equations (a.7)-(a.8), we obtain a nonlinear autoregressive representation: $Y_T = a( Y_{T-1}, v_T)$, where the $v_T$'s are i.i.d. such that $(v_{1,T}), (v_{2,T})$ are independent. \\

Alternatively, one can use the ordering: $Y_{2,T}$ followed by $Y_{1,T}$ given $Y_{2,T}$. More generally, for any invertible nonlinear transformation $Y_T^* = c(Y_T)$, the above approach can be applied first to $Y_{1,T}^*$ and next to  $Y_{2,T}^*$ conditional on  $Y_{1,T}^*$. Therefore any Markov process can be written as a nonlinear causal autoregressive 
process and the above discussion shows that this autoregressive representation is not unique.\\

 \textbf{b) Identification of the nonlinear causal autoregressive representation}\\

It is equivalent to consider the identification of function $a$ or the identification of nonlinear innovations.
Let us now describe in detail all the nonlinear causal innovation identification issues. First, we can assume that $v_{1,T}, v_{2,T}$
are i.i.d. and independent of one another with uniform distributions on [0,1]. We need to find out if there exists another pair of variables $w_{1,T}, w_{2,T}$, which are independent  and uniformly distributed such that:

$$a(Y_{T-1}, w_T) = a(Y_{T-1}, v_T), \;\; \forall\;Y_{T-1},$$

 or, equivalently, a pair of variables $w_T$ that satisfy a (nonlinear) one-to-one relationship with $v_T$. Let $w=b(v)$ denote this relationship. We have the following Lemma:

\begin{lemma}
	Let us assume that $b$ is continuous, twice differentiable and that the Jacobian matrix $\partial b(v)/\partial v'$ has distinct eigenvalues. Then, the components of $b$ are harmonic functions, that is:
	
	$$\frac{\partial^2 b_j(v)}{\partial v_1^2} + \frac{\partial^2 b_j(v)}{\partial v_2^2} = 0, \;\; j=1,2.$$
\end{lemma}

 {\bf Proof:} i) We can apply the Jacobian formula to get the density of $w$ given the density of $v$. Since both joint densities are uniform, it follows that $| \det  \frac{\partial b(v)}{\partial v'}| = 1$,
$\forall v \in [0,1]^2$. \\

ii) Let us consider the eigenvalues $\lambda_1 (v), \lambda_2(v)$ of the Jacobian matrix $ \frac{\partial b(v)}{\partial v'}$. The eigenvalues are continuous functions of this matrix, and therefore continuous functions of $v$
(whenever these eigenvalues are different). Then, two cases can be distinguished:\\
\begin{itemize}
	\item Case 1: The eigenvalues are real.
	\item	Case 2: The eigenvalues are complex conjugates.
\end{itemize}

 iii) In case 1, we have $\lambda_2 (v) = 1/\lambda_1(v)$ ( or $-1/\lambda_1(v)$), where $ \lambda_1(v)$ is less or equal to 1 in absolute value for any $v$, and then $\lambda_2(v)$ is larger than or equal to 1 in absolute value for any $v$. Since $b(v) \in [0,1]^2$ for any $v \in [0,1]^2$, it follows that $\lambda_2(v)$ cannot be explosive.\\

iv) Therefore case 2 of complex conjugate roots is the only relevant one. Let us consider the case $\det  \frac{\partial b(v)}{\partial v'}=1$, $\forall v \in [0,1]^2$ (the analysis of $\det  \frac{\partial b(v)}{\partial v'} = -1$ is similar).
Then, the Jacobian matrix is a rotation matrix:\\

$$\frac{\partial b(v)}{\partial v'}  = \left(\begin{array}{cc} \frac{\partial b_1(v)}{\partial v_1} & \frac{\partial a_1(v)}{\partial v_2} \\ \frac{\partial b_2(v)}{\partial v_1} & \frac{\partial b_2(v)}{\partial v_2}  \end{array}\right)
\equiv \left(\begin{array}{cc} \cos \theta(v) & - \sin \theta (v) \\ \sin \theta(v) & \cos \theta(v)  \end{array}\right).
$$

Thus the standard identification issue known in the linear SVAR model, that is up to a given rotation (i.e. orthogonal) matrix, is replaced by the analogue in which the rotation matrix is local and depends on $v$. We deduce that:
\begin{eqnarray}\label{a4}
\frac{\partial b_1(v)}{\partial v_1} & = & \frac{\partial b_2(v)}{\partial v_2}, \nonumber \\
\frac{\partial b_1(v)}{\partial v_2} & = & - \frac{\partial b_2(v)}{\partial v_1}.
\end{eqnarray}
Let us differentiate the first equation with respect to $v_1$ and the second one with respect to $v_2$. We get:
\begin{equation}
\frac{\partial^2 b_1(v)}{\partial v_1^2} = \frac{\partial^2 b_2(v)}{\partial v_1 \partial v_2}\;\; \mbox{and} \;\;
\frac{\partial^2 b_1(v)}{\partial v_2^2} = - \frac{\partial^2 b_2(v)}{\partial v_1 \partial v_2},
\end{equation}

 and by adding these equalities:

\begin{equation}\label{a6}
\frac{\partial^2 b_1(v)}{\partial v_1^2} + \frac{\partial^2 b_1(v)}{\partial v_2^2} = 0.
\end{equation}

 Therefore $b_1$ is a harmonic function that satisfies the Laplace equation \eqref{a6}. Similarly, $b_2$ is also a harmonic function. \qed \\

Harmonic functions are regular functions: they are infinitely differentiable and have series representations that can be differentiated term by term [Axler et al. (2001)]:

\begin{eqnarray}\label{a9}
b_1(v) & = & \sum_{h=0}^{\infty} \sum_{k=0}^{\infty} (b_{1hk} v_1^k v_2^h), \nonumber \\
b_2(v) & = & \sum_{h=0}^{\infty} \sum_{k=0}^{\infty} (b_{2hk} v_1^k v_2^h).
\end{eqnarray}

 Moreover, these series representations are unique. Then, we can apply the conditions \eqref{a4} to these expansions to derive the constraints on the series coefficients and the link between functions $b_1$ and $b_2$.\\

Let us define:
$$\frac{\partial b_1(v)}{\partial v_1} =  \frac{\partial b_2(v)}{\partial v_2} \equiv \sum_{h=0}^{\infty}
\sum_{k=0}^{\infty} (c_{hk} v_1^h v_2^k).$$
 Then, by integration, we get:
\begin{eqnarray*}
b_1(v) & \equiv & \sum_{h=0}^{\infty} \sum_{k=0}^{\infty} [ c_{hk} \frac{v_1^{h+1}}{h+1} v_2^k] + \sum_{k=0}^{\infty}
d_{1k} v_2^k, \\
b_2 (v) & \equiv & \sum_{h=0}^{\infty} \sum_{k=0}^{\infty} [ c_{hk} v_1^h \frac{v_2^{k+1}}{k+1}] + \sum_{h=0}^{\infty}
d_{2h} v_1^h,
\end{eqnarray*}
 where the second sums on the right hand sides are the integration "constants". Equivalently, we have:
\begin{eqnarray*}
b_1(v) & \equiv & \sum_{k=0}^{\infty}
d_{1k} v_2^k +  \sum_{h=1}^{\infty} \sum_{k=0}^{\infty} [ c_{h-1,k} \frac{v_1^{h}}{h} v_2^k], \\
b_2 (v) & \equiv & \sum_{h=0}^{\infty}
d_{2h} v_1^h  + \sum_{h=0}^{\infty} \sum_{k=1}^{\infty} [ c_{h,k-1} v_1^h \frac{v_2^{k}}{k}].
\end{eqnarray*}
Let us now write the second equality in \eqref{a4}, i.e.
$$\frac{\partial b_1 (v)}{\partial v_2}  = - \frac{\partial b_2 (v)}{\partial v_1}.$$
This yields:
\begin{eqnarray}\label{A8}
\frac{k+1}{h} c_{h-1, k+1} & = & - \frac{h+1}{k}  c_{h+1, k-1}, \;\; h \geq 1, k  \geq 1,  \\
\frac{1}{h} c_{h-1, 1} & = & - (h+1) d_{2,h+1}  \;\; h \geq 1, \nonumber \\
\frac{1}{h} c_{1, k-1} & = & - (k+1) d_{1,k+1}  \;\; k \geq 1, \nonumber \\
d_{11} & = & - d_{21}. \nonumber
\end{eqnarray}
The set of restrictions \eqref{A8} provides information on the dimension of underidentification. As the dimension concerns functional spaces, we describe it from the series expansions \eqref{a9} and the number of independent
parameters $b_{1, h, k}, b_{2, h, k}$ with $h+k \leq m$. This number is equal to $(m+1)(m+2)/2$. Let us now prove Proposition 8. 

\begin{lemma}
	The space of parameters $(b_{1, h, k}, b_{2, h, k},\; h+k \leq m)$ is of dimension $2m$.
\end{lemma}
 {\bf Proof:} Let us consider an alternative parametrization with parameters $c_{h,k}, d_{1,h}, d_{2,h}$. The parameters $b_{1, h, k}, b_{2, h, k}$ with $h+k=j$ are linear functions of parameters $c_{h,k}, h+k = j+1, d_{1,j+1}d_{2,j+1}$. Then the result result follows from restrictions \eqref{A8}. \qed\\

Other identification issues can arise if transformation $b$ is not assumed twice continuously differentiable. Let us consider the first component $v_1$ that follows the uniform distribution on $[0,1]$ and introduce two intervals $[0,c]$ and $[1-c,1]$ with $c<0.5$. Then, the variable $w_1$ defined by:
$$ w_1 = \left\{ \begin{array}{l} v_1, \; \mbox{if} \; v_1 \in (c, 1-c), \\
2 v_1 - 1, \; \mbox{if} \; v_1 \in (0,c) \cup (1-c,1),
\end{array} \right.
$$
also follows the uniform distribution and, similarly to $v_1$, variable $w_1$ is independent of $v_2 = w_2$. Note that this transformation is not monotonous. Therefore, the size $\delta$ of a shock to $v_1$ is difficult to interpret in terms of a magnitude of a shock to $w_1$. We conclude that, contrary to the linear framework, in a nonlinear dynamic framework the assumption of independence between the components of $v_t$ is insufficient to identify the structural innovations to be shocked.

\subsection{Proof of Proposition 8}

Let us consider the linear transformation $A$ in $y_t=Ax_t$. Since the components of $x_t$ are defined up to a scale factor, we can assume $A=\begin{bmatrix}
	1 \ a_{12} \\ 
	a_{21} \ 1 \\ 
\end{bmatrix}$, with $a_{12}a_{21}\neq1$ to ensure that $A$ is invertible. Then, we can compute the cross ACF of the process $y_t$: $\Gamma(h) = \begin{bmatrix}
	\gamma_{11}(h) \ \gamma_{12}(h) \\
	\gamma_{21}(h) \ \gamma_{22}(h) \\
\end{bmatrix}$, in terms of the marginal ACFs $\tilde{\gamma}_1(h),\tilde{\gamma}_2(h)$ of $(x_{1,t})$ and $(x_{2,t})$, respectively. We get:
\begin{equation*}
	\begin{split}
		\gamma_{11}(h) & = \tilde{\gamma}_1(h) + a_{12}^2 \tilde{\gamma}_2(h), \\
		\gamma_{12}(h) & = a_{21}\tilde{\gamma}_1(h) + a_{12}\tilde{\gamma}_2(h), \\
		\gamma_{22}(h) & = a_{21}^2\tilde{\gamma}_1(h) +  \tilde{\gamma}_2(h). \\
	\end{split}
\end{equation*}	
The system above can be solved to derive the relationship between the observable autocovariances. We get:
\begin{equation}
	\gamma_{12}(h)  = \frac{1}{1+a_{12}a_{21}}[a_{21}\gamma_{11}(h)+a_{12}\gamma_{22}(h)], \ \text{(if $a_{12}a_{21}\neq -1$)}.
\end{equation}
Therefore, if the ACF's $\gamma_{11}(h)$, $\gamma_{22}(h)$ are linearly independent, we can identify $\frac{a_{12}}{(1+a_{12}a_{21})}$ and $\frac{a_{21}}{(1+a_{12}a_{21})}$, or equivalently $c=\frac{a_{21}}{a_{12}}$ and $d=\frac{a_{21}}{(1+a_{12}a_{21})}=\frac{ca_{12}}{(1+ca^2_{12})}$. We deduce that $a_{12}$ is a solution of the equation of degree 2: 
\begin{equation*}
	d c a_{12}^2 -ca_{12} + d = 0, 
\end{equation*}
that has two solutions. Matrix $A$ is locally identifiable, not globally identifiable. But this lack of global identification is due to the definition of $A$, up to the permutation of components $x_{1,t}$ and $x_{2,t}$. Therefore, $A$ is essentially unique. Then, we can identify the sources by inverting the matrix A as $x_t=A^{-1}y_t$, and also the nonlinear dynamics of $x_{1,t}$ and $x_{2,t}$, respectively, since these Markov processes are one dimensional. 

\subsection{Proof of Proposition 9}

\textbf{(i) Transition of $(y_t)$ at horizon $h$.}\\

Let us assume $y_t = A(x_t) \iff x_t = C(y_t)$, where $C$ is the inverse of function $A$, that is the demixing function, and denote $\ell_{1,h}(x_{1,t+h}|x_{1,t}),\ell_{2,h}(x_{2,t+h}|x_{2,t})$ the transition densities at horizon $h$ of processes $(x_{1,t})$ and $(x_{2,t})$, respectively. Then the transition at horizon $h$ of Markov process $(y_t)$ is given by: 
\begin{equation}
	\ell^*_{h}(y_{t+h}|y_t) = \left|\det\frac{\partial C(y_{t+h})}{\partial y'}\right|\ell_{1,h}[c_1(y_{t+h})|c_1(y_t)]\ell_{2,h}[c_2(y_{t+h})|c_2(y_t)],
\end{equation}
by the Jacobian formula, and by taking the logarithm of both sides:
\begin{equation}
	\log	\ell^*_{h}(y|z) = \log\left|\det\left(\frac{\partial C(y)}{\partial y'}\right)\right|+\log\ell_{1,h}[c_1(y)|c_1(z)]+\log\ell_{2,h}[c_2(y)|c_2(z)],
\end{equation}
valid for $ y,z \in [0,1]^2, \forall h = 1,2,...$. We get an infinite set of relations, in which the left hand side is nonparametrically identifiable. \\ 

\textbf{(ii) Identifiability of the Jacobian} \\ 

When $h$ tends to infinity, the right hand side tends to $\log \left|\det \left(\frac{\partial C(y)}{\partial y'}\right)\right|$. Indeed, by the ergodicity of process $(y_t)$, the conditional transitions $\ell_{1,h}$,$\ell_{2,h}$ no longer depend on the conditioning values. Thus, they coincide with the unconditional density, that is equal to 1, since $x_1$ and $x_2$ have been normalized to have marginal distributions that are uniform on $[0,1]$. \\ 

\textbf{(iii) Identifiability of C (up to a linear affine transformation).}\\

Therefore, we can identify:
\begin{equation}
	\Lambda_h(y;z) = \log \ell_{1,h}[c_1(y)|c_1(z)] + \ell_{2,h}[c_2(y)|c_2(z)], \forall \ y,z,h.
\end{equation}
Let us take the derivative of these equations with respect to y. Then we get:
\begin{equation}
	\frac{\partial 	\Lambda_h(y;z)  }{\partial y} = \frac{\partial c_1(y)}{\partial y} \frac{\partial\log \ell_{1,h}}{\partial x_1} [c_1(y)|c_1(z)] + \frac{\partial c_2(y)}{\partial y} \frac{\partial\log \ell_{2,h}}{\partial x_2} [c_2(y)|c_2(z)] , \forall \ y,z,h. 
\end{equation}
Under the condition in Proposition 8, the sequences $\left[\frac{\partial \log \ell_{jh} }{\partial u}(u|v), \ \text{v varying}\right]$, $j=1,2$, are linearly independent. Then, the sequence $\left[\frac{\partial\Lambda_h(\cdot;v) }{\partial y}, \ \text{v varying}\right]$ generates a subspace of dimension 2, in which the two sequences above form a basis. Therefore, there exists a (2,2) matrix $D$ such that: $\frac{\partial C(y)}{\partial y'}D$ is identifiable. By integrating with respect to $y$, we deduce that the diffeomorphism $C(\cdot)$ is identifiable up to a linear affine transformation. \\ 

\textbf{(iv) Identification of the sources $x_t$ and their dynamics.} \\ 

The end of the proof is similar to the end of the proof in Proposition 7. 

\section*{Online Appendices}

\subsection*{B.1 The Set of Linear Affine Transformations with Jacobian Equal to One for $n=2$}

%\begin{figure}[h]
%	\centering
%	\includegraphics[width=0.7\linewidth]{diff}
%	\caption{A diffeomorphism. [Source: Wikimedia Commons]}
%	\label{fig:diff}
%\end{figure}

To find the transformation of $u=(u_1,u_2)$, we have just to consider the crossing of the almost vertical curves starting from $u_1$ and the almost horizontal curve starting from $u_2$. The crossing point is $\tilde{u}_1,\tilde{u}_2$, with $\tilde{u}=(\tilde{u}_1,\tilde{u}_2)' = T(u)$. There are no transformations on the boundaries to ensure the marginal uniform distribution conditions and no discontinuity when transforming the corners. We start by considering the restricted sets of linear, or linear affine transformations before providing more general results.

\subsubsection*{Linear Transformations}

Let us focus on the linear transformation: $T(u)=Tu$, where $T$ is a $(2,2)$ matrix. It is easily checked hat the condition $Tu\in[0,1]^2$, $\forall u \in [0,1]^2$ implies:
\begin{equation*}
	T_{i,j} \geq 0, \forall i,j=1,2, \ \text{and} \ \sum_{j=1}^{2}T_{i,j} \leq 1, \forall i = 1,2. 
\end{equation*}
Then, by the Perron-Froebenius Theorem on nonnegative matrices, the eigenvalue of matrix $T$ with the largest modulus is real positive smaller or equal to one. The condition $\det T=+1$ implies that the product of the modulus of eignvalues of $T$ is equal to one. This implies that all eigenvalues are such that $|\lambda_j|=1,j=1,2$ the largest eigenvalue being $\lambda_1=1$. Another application of the Perron-Frobenius theorem implies that the other eigenvalue is also equal to one. \\

\begin{proposition}
	The identity is the only admissible linear transformation $T$. 
\end{proposition}

\textbf{Proof:} If $T$ is diagonalizable, $T$ is the identity. Let us now show that $T$ is necessarily diagonalizable. Otherwise, $T$ can be written under a triangular form by Jordan representation: $T=Q\begin{bmatrix}
	1& 1 \\
	0 &1 \\ 
\end{bmatrix}Q^{-1}$. Then, we have $T^h =Q\begin{bmatrix}
	1 & 1 \\
	0 & 1 \\ 
\end{bmatrix}^hQ^{-1}=Q\begin{bmatrix}
	1&  h \\
	0& 1 \\ 
\end{bmatrix}Q^{-1}$, with elements tending to infinity of h tends to infinity. This contradicts the constraint: $T^h([0,1]^2)=([0,1])^2$. \qed

\subsubsection*{Affine Transformations, $n=2$}
The linear transformations considered in Section 2.1 above do not account for the possibility of recovering a uniform distribution on $[0,1]$ by the mapping: $u \rightarrow 1-u$. Let us now consider affine transformations of the type:
\begin{equation*}
	T(u) = \begin{bmatrix}
		u_1 - 1/2 \\
		u_2 - 1/2 \\ 
	\end{bmatrix},
\end{equation*}
written for $n=2$. The condition $T(u)\in[0,1]^2$ is equivalent to the inequality restrictions:
\begin{equation*}
	\begin{split}
		-1 \leq &  T_{11}+T_{12}\leq 1, \ -1\leq T_{11}-T_{12}\leq 1, \\
		-1 \leq &  T_{21}+T_{22}\leq 1, \ -1\leq T_{21}-T_{22}\leq 1, \\ 
	\end{split}
\end{equation*}
and the Jacobian condition (with value $+1$) gives the equality:
\begin{equation*}
	T_{11}T_{22}-T_{12}T_{21}=1.
\end{equation*}
We would have now to look for the matrices $T$ satisfying jointly all the above restrictions.

\subsection*{B.3 Diffeomorphisms}

The normally distributed innovations $(\varepsilon_{1,t},\varepsilon_{2,t})$ can be written on polar coordinates such that: 
\begin{equation*}
	\begin{cases}
		\varepsilon_{1,t} = \rho \cos \theta,\\
		\varepsilon_{2,t} = \rho \sin \theta. \\ 
	\end{cases} \iff 	\begin{cases}
		\rho = \sqrt{\varepsilon_{1,t}^2+\varepsilon_{2,t}^2},\\
		\theta= \tan^{-1}(\varepsilon_{2,t}/\varepsilon_{1,t}). \\ 
	\end{cases}
	%	\begin{split}
		%		\varepsilon_{1,t} = \rho \cos \theta,\\
		%		\varepsilon_{2,t} = \rho \sin \theta. \\ 
		%	\end{split}
\end{equation*}
Then, the density of the standard normal distribution becomes $f(\rho,\theta)=\frac{1}{2\pi}\rho\exp(-\rho^2/2)$, on the domain $(0,\infty)\times(0,2\pi)$ (up to modulo $2\pi$ for $\theta$). Now, consider the transformation $(\rho,\theta) \rightarrow (\rho^*,\theta^*) = (\rho,\theta+a(\rho))$, where $a(\cdot)$ is a continuous differentiable function of $\rho$\footnote{Note that the Jacobian formula is  valid for functions that are diffeomorphisms almost everywhere. In our framework, the transformation is a piecewise diffeomorphism and has to be understood with $\theta + a(\rho)$, up to modulo $2\pi$.}. Then we have: 
\begin{equation*}
	\begin{cases}
		\varepsilon^*_{1,t} = \rho^* \cos \theta^*,\\
		\varepsilon^*_{2,t} = \rho^* \sin \theta^*. \\ 
	\end{cases} \iff 	\begin{cases}
		\rho^* = \sqrt{\varepsilon_{1,t}^2+\varepsilon_{2,t}^2},\\
		\theta^*= \tan^{-1}(\varepsilon_{2,t}/\varepsilon_{1,t})+a\left(\sqrt{\varepsilon_{1,t}^2+\varepsilon_{2,t}^2}\right). \\ 
	\end{cases}
	%	\begin{split}
		%		\varepsilon_{1,t} = \rho \cos \theta,\\
		%		\varepsilon_{2,t} = \rho \sin \theta. \\ 
		%	\end{split}
\end{equation*}
Thus, we get for $(u^*_{1,t},u^*_{2,t})=(\Phi(u_{1,t}),\Phi(u_{2,t}))$:
\begin{equation}
	\begin{cases}
u^*_{1,t} = \Phi \left[\sqrt{\Phi^{-1}(u_{1,t})^2+\Phi^{-1}(u_{2,t})^2} \cos\left(\tan^{-1}(\Phi^{-1}(u_{2,t})/\Phi^{-1}(u_{1,t}))+a\left(\sqrt{\Phi^{-1}(u_{1,t})^2+\Phi^{-1}(u_{2,t})^2}\right)\right)\right] ,\\
u^*_{2,t} = \Phi \left[\sqrt{\Phi^{-1}(u_{1,t})^2+\Phi^{-1}(u_{2,t})^2} \sin\left(\tan^{-1}(\Phi^{-1}(u_{2,t})/\Phi^{-1}(u_{1,t}))+a\left(\sqrt{\Phi^{-1}(u_{1,t})^2+\Phi^{-1}(u_{2,t})^2}\right)\right)\right] . \\ 
\end{cases}
\end{equation}

\end{appendices}

\end{document}